\documentclass[USenglish,cleveref, autoref, thm-restate]{theoretics}

\pdfoutput=1 

\usepackage[]{todonotes}

\newcommand{\nocontentsline}[3]{}
\newcommand{\tocless}[2]{\bgroup\let\addcontentsline=\nocontentsline#1{#2}\egroup}

\newcommand{\defin}[1]{\textit{#1}}

\renewcommand{\and}{\text{ and }}

\newcommand{\compl}[1]{\,^\mathrm{c} #1}

\newcommand{\cG}{\mathcal G}
\newcommand{\cS}{\mathcal S}
\newcommand{\cP}{\mathcal P}
\newcommand{\col}{\mathrm{col}}

\newcommand{\Eve}{\mathrm{Eve}}
\newcommand{\Adam}{\mathrm{Adam}}
\newcommand{\Parity}{\mathrm{Parity}}

\newcommand{\val}{\mathrm{val}}

\newcommand{\Safety}{\mathrm{Safety}}
\newcommand{\Reachability}{\mathrm{Reachability}}

\newcommand{\Z}{\mathbb Z}
\newcommand{\N}{\mathbb N}

\newcommand{\game}{\mathcal{G}}

\newcommand{\VE}{V_{\text{Eve}}}
\newcommand{\VA}{V_{\text{Adam}}}

\newcommand{\lex}{\mathrm{lex}}

\newcommand{\powerset}[1]{\mathcal P(#1)}
\newcommand{\nepowerset}[1]{\mathcal P^{\neq \varnothing}(#1)}

\renewcommand{\L}{\mathcal L}

\renewcommand{\path}[2]{\Pi_{#1}^{#2}}
\newcommand{\ipath}[1]{\path{#1}{\infty}}

\newcommand{\re}[1]{\xrightarrow{#1}}
\newcommand{\rp}[1]{\overset{#1}{\rightsquigarrow}}

\newcommand{\tin}{\text{ in }}
\newcommand{\tif}{\text{ if }}
\newcommand{\tand}{\text{ and }}
\newcommand{\tor}{\text{ or }}

\newcommand{\tred}{\mathrm{red}}
\newcommand{\tblue}{\mathrm{blue}}

\newcommand{\tgray}{\mathrm{gray}}

\renewcommand{\S}{\mathcal S}

\newcommand{\safe}{\mathtt{safe}}
\newcommand{\bad}{\mathtt{bad}}
\newcommand{\imm}{\mathtt{imm}}

\newcommand{\good}{\mathtt{good}}
\newcommand{\wait}{\mathtt{wait}}

\newcommand{\Buchi}{\mathrm{\text{B\"uchi}}} 
\newcommand{\Cobuchi}{\mathrm{Co}\text{-}\mathrm{\text{B\"uchi}}} 

\newcommand{\req}{\mathtt{req}}

\newcommand{\Energy}{\mathrm{Energy}}

\newcommand{\BoundedCounter}{\mathrm{Bounded}}

\newcommand{\eps}{\varepsilon}

\newcommand{\QEnergy}{\mathrm{QEnergy}}

\renewcommand{\req}{\mathrm{req}}

\newcommand{\Grant}{\mathrm{Grant}}
\newcommand{\FinParity}{\mathrm{FinParity}}

\renewcommand{\lex}{\otimes}

\newcommand{\TW}{\mathrm{TW}}
\newcommand{\TL}{\mathrm{TL}}

\newcommand{\W}{\mathcal{W}}

\newcommand{\fin}{\mathrm{fin}}

\newcommand{\Pref}{\mathrm{Pref}}

\renewcommand{\H}{\mathcal{H}}

\newcommand{\emptyword}{\epsilon}

\newcommand{\prev}{\mathrm{prev}}

\newcommand{\li}[1]{\overline{#1}}

\title{Characterizing Positionality in Games of Infinite Duration over Infinite Graphs} 

\ThCSauthor{Pierre Ohlmann}{pohlmann@mimuw.edu.pl}[https://orcid.org/0000-0002-4685-5253]
\ThCSaffil{University of Warsaw, Poland}{}

\ThCSyear{2023}
\ThCSarticlenum{3}
\ThCSreceived{Jun 22, 2022}
\ThCSaccepted{Dec 23, 2022}
\ThCSpublished{Jan 30, 2023}
\ThCSkeywords{infinite duration games, positionality, universal graphs}
\ThCSdoi{10.46298/theoretics.23.3}
\ThCSshortnames{P. Ohlmann} 
\ThCSshorttitle{Positionality in Infinite Games} 
\ThCSthanks{I want to thank Antonio Casares, Thomas Colcombet, Nathana\"el Fijalkow and Pierre Vandenhove for so many insightful discussions on the topic.}

\DeclareUnicodeCharacter{0301}{XXXXXXXXXXXXX}

\begin{document}

\maketitle

\begin{abstract}
    We study turn-based quantitative games of infinite duration opposing two antagonistic players and played over graphs.
    This model is widely accepted as providing the adequate framework for formalizing the synthesis question for reactive systems.
    This important application motivates the question of
    strategy complexity: which valuations (or payoff functions) admit optimal positional strategies (without memory)?
    Valuations for which both players have optimal positional strategies have been characterized by Gimbert and Zielonka~\cite{GZ05} for finite graphs and by Colcombet and Niwi\'nski~\cite{CN06} for infinite graphs.
    
    However, for reactive synthesis, existence of optimal positional strategies for the opponent (which models an antagonistic environment) is irrelevant.
    Despite this fact, not much is known about valuations for which the protagonist admits optimal positional strategies, regardless of the opponent.
    In this work, we characterize valuations which admit such strategies over infinite game graphs.
    Our characterization uses the vocabulary of universal graphs, which has also proved useful in understanding recent breakthrough results regarding the complexity of parity games.
    
    More precisely, we show that a valuation admitting universal graphs which are monotone and well-ordered is positional over all game graphs, and -- more surprisingly -- that the converse is also true for valuations admitting neutral colors.
    We prove the applicability and elegance of the framework by unifying a number of known positionality results, proving new ones, and establishing closure under lexicographical products.
    Finally, we discuss a class of prefix-independent positional objectives which is closed under countable unions.
\end{abstract}

\section{Introduction}

\paragraph{Games.}
In zero-sum turn-based infinite duration games played on graphs, two players, Eve and Adam, take turns in moving a token along the edges of a given (potentially infinite) directed graph, whose edges have labels from a given set of colors.
This interaction goes on forever and produces an infinite sequence of colors according to which the outcome of the game is determined, using a valuation which is fixed in advance.
The complexity of a strategy for either of the two players can be measured by means of how many states are required to implement it as a transition system.
In this paper, we are interested in the question of positionality (which corresponds to the degenerate case of memory one) for Eve\footnote{Some authors use the term ``half-positionality'' to refer to what we will simply call ``positionality''.}: for which valuations is it the case that Eve can play optimally without any memory, meaning that moves depend only on the current vertex of the game, regardless of the history leading to that vertex.

Understanding memory requirements -- and in particular positionality -- of given valuations has been a deep and challenging endeavour initiated by Shapley~\cite{Shapley53} for finite concurrent stochastic games, and then in our setting by Büchi and Landweber~\cite{BL69}, Büchi~\cite{Buchi77} and Gurevich and Harrington~\cite{GH82}.
The seminal works of Shapley~\cite{Shapley53}, Ehrenfeucht and Mycielski~\cite{EM73}, and later Emerson and Jutla~\cite{EJ91}, Klarlund~\cite{Klarlund92}, McNaughton~\cite{McNaughton93} and Zielonka~\cite{Zielonka98}, have given us a diverse set of tools for studying these questions.

Roughly speaking, these early efforts culminated in Gimbert and Zielonka's~\cite{GZ05} characterization of bi-positionality (positionality for both players) over finite graphs on one hand, and Kopczy\'{n}ski's~\cite{Kopczynski06} results and conjectures on positionality on the other.
In the recent years, increasingly expressive and diverse valuations have emerged from the development of reactive synthesis, triggering more and more interest in these questions.

As we will see below, bi-positionality is by now quite well understood, and the frontiers of finite-memory determinacy are becoming clearer.
However, recent approaches to finite-memory determinacy behave badly when instantiated to the case of positionality, for different reasons which are detailed below.
Therefore, inspired by the works of Klarlund, Kopczy\'{n}ski and others, we propose a general framework for positionality, and establish a new characterization result.
Before introducing our approach, we briefly survey the state of the art, with a focus on integrating several recent and successful works from different broadly related settings.

\paragraph{Bi-positionality.}
The celebrated result of Gimbert and Zielonka~\cite{GZ05} characterizes valuations which are bi-positional over finite graphs (including parity objectives, mean-payoff, energy, and discounted valuations, and many more).
The characterization is most useful when stated as follows (one-to-two player lift): a valuation is bi-positional if (and only if) each player has optimal positional strategies on game-graphs which they fully control.
Bi-positionality over infinite graphs is also well understood thanks to the work of Colcombet and Niwi\'nski~\cite{CN06}, who established that any prefix-independent objective which is bi-positional over arbitrary graphs is, up to renaming the colors, a parity condition (with finitely many priorities).

\paragraph{Finite-memory determinacy.}
Finite-memory determinacy of Muller games over finite graphs was first established by Büchi and Landweber~\cite{BL69}, and the result was extended to infinite graphs by Gurevich and Harrington~\cite{GH82}.
Zielonka~\cite{Zielonka98} was the first to investigate precise memory requirements and he introduced what Dziembowski, Jurdzi\'{n}ski and Walukiewicz~\cite{DJW97} later called the Zielonka tree of a given Muller condition, a data structure which they used to precisely characterize the amount of memory required by optimal strategies.

Another precise characterization of finite memory requirements was given by Colcombet, Fijalkow and Horn~\cite{CFH14} for generalised safety conditions over graphs of finite degree, which are those defined by excluding an arbitrary set of prefixes (topologically, $\Pi_1$).
This characterization is orthogonal to the one for Muller conditions (which are prefix-independent); it provides in particular a proof of positionality for generalisations of (threshold) energy objectives, and different other results.

Le Roux, Pauly and Randour~\cite{LPR18} identified a sufficient condition ensuring that finite memory determinacy (for both players) over finite graphs is preserved under boolean combinations.
Although they encompass numerous cases from the literature, the obtained bounds are generally not tight, and thus their results instantiate badly to the case of positionality.

We mention also a recent general result of Bouyer, Le Roux and Thomasset~\cite{BLT22}, in the much more general setting of (graph-less) concurrent games given by a condition $W \subseteq {(A \times B)}^\omega$: if $W$ belongs to $\Delta_2^0$ and residuals form a well-quasi order
, then it is finite-memory determined\footnote{In the concurrent setting, games are often not even determined (even when Borel). This is not an issue for considering finite-memory determinacy, which means ``if a winning strategy exists, then there is one with finite memory''.}.
We will also rely on well-founded orders (although ours are total), but stress that our results are incomparable: to transfer the result of~\cite{BLT22} to game on graphs, one encodes the (possibly infinite) graph in the winning condition $W$, and therefore strategies with reduced memory no longer have access to it.
This gives finite memory determinacy if the graph is finite (and if one complies with having memory bounds depending on its size), however positionality results cannot be transferred.

\paragraph{Chromatic and arena-independent memories.}
In his thesis, Kopczy\'{n}ski~\cite{Kopczynski06} proposed to consider strategies implemented by memory-structures that depend only on the colors seen so far (rather than on the path), which he called chromatic memory -- as opposed to usual chaotic memory.
His motivations for studying chromatic memory are the following: first, it appears that for several (non-trivial) conditions, chromatic and chaotic memory requirements actually match; second, any $\omega$-regular condition $W$ admits optimal strategies with finite chromatic memory, implemented by a deterministic (parity or Rabin) automaton recognising $W$; third, such strategies are arena-independent, and one may even prove (Proposition 8.9 in~\cite{Kopczynski06}) that in general, there are chromatic memories of minimal size which are arena-independent.
Kopczy\'nski therefore poses the following question: does it hold that chromatic (or equivalently, arena-independent) and chaotic memory requirements match in general?

A recent work of Casares~\cite{Casares22} studies this question specifically for Muller games, for which an elegant characterization of chromatic memory is given: it coincides with the size of the minimal deterministic transition-colored Rabin automaton recognising it.
Comparing with the characterization of~\cite{DJW97} via Zielonka trees reveals a gap between arena-dependent and independent memory requirements already for Muller conditions, which answers the above question in the negative.

Arena-independent (finite) memory structures have also been investigated recently by Bouyer, Le Roux, Oualhadj, Randour and Vandenhoven~\cite{BLORV22} over finite graphs.
In this context, they were able to generalise the characterization of~\cite{GZ05} (which corresponds to memory one), to arbitrary memory structures.
As a striking consequence, the one-to-two player lift of~\cite{GZ05} extends to arena-independent finite memory: if both players can play optimally with finite arena-independent memory respectively $n_\Eve$ and $n_\Adam$ in one-player arenas, then they can play optimally with finite arena-independent memory $n_\Eve \cdot n_\Adam$ in general.
A counterexample is also given in~\cite{BLORV22} for one-to-two player lifts in the case of arena-dependent finite memory.

This characterization was more recently generalised to pure arena-independent strategies in stochastic games by Bouyer, Oualhadj, Randour and Vandenhoven~\cite{BORV21}, and even to concurrent games on graphs by Bordais, Bouyer and Le Roux~\cite{BBL21}.
Unfortunately, none of these result carry over well to positionality, since they inherit from~\cite{GZ05} the requirement that both players rely on the same memory structure.
For instance, in a Rabin game, the antagonist requires finite memory $>1$ in general, and therefore the results of~\cite{BLORV22} cannot establish positionality.
We also mention a very recent work of Bouyer, Randour and Vandehoven~\cite{BRV22} which establishes that the existence of optimal finite chromatic memory for both players over arbitrary graphs characterizes $\omega$-regularity of the objective.

\paragraph{Positionality.}
Unfortunately, there appears to be no characterization similar to Gimbert and Zielonka's for (one player) positionality.
In fact, there has not been much progress in the general study of positionality since Kopczy\'nski's work, on which we now briefly extend.

Kopczy\'nski's main conjecture~\cite{Kopczynski06} on positionality asserts that prefix-independent positional objectives are closed under finite unions\footnote{This immediately fails for bi-positionality; for instance, the union of two co-B\"uchi objectives is not positional for the opponent.}. 
It can be instantiated either for positionality over arbitrary graphs, or only finite graphs, leading to two incomparable variants.
The variant over finite game graphs was recently disproved by Kozachinskiy~\cite{Kozachinskiy22}, however the question remains open for infinite game graphs (which are the focus of the present paper).

An elegant counterexample to a stronger statement is presented in~\cite{Kopczynski06}: there are uncountable unions of Büchi conditions which are not positional over some countable graphs.
One of Kopczy\'nski's contributions lies in introducing two classes of prefix-independent objectives, concave objectives and monotone objectives, which are positional (over finite and arbitrary graphs, respectively) and closed under finite unions.
%

Monotone objectives are those of the form $C^\omega \setminus \L^\omega$, where $\L \subseteq C^*$ is a (regular) language recognized by a linearly ordered deterministic automaton\footnote{The automaton is assumed to be finite, but Kopczy\'nski points out (page 45 in~\cite{Kopczynski06}) that the main results still hold whenever the state space is well-ordered and admits a maximum (stated differently, it is a non-limit ordinal).} whose transitions are monotone.
Our work builds on Kopczy\'nski's suggestion to consider well-ordered monotone automata; however to obtain a complete characterization we make several adjustments, most crucially we replace the automata-theoretic semantic of recognisability by the graph-theoretical universality which is more adapted to the fixpoint approach we pursue.

\paragraph{Our approach.}
We introduce well-monotone graphs, which are well-ordered graphs over which each edge relation is monotone, and prove in a general setting that existence of universal well-monotone graphs implies positionality.
The idea of using adequate well-founded (or ordinal) measures to fold arbitrary strategies into positional ones is far from being novel: it appears in the works of Emerson and Jutla~\cite{EJ91} (see also Walukiewicz' presentation~\cite{Walukiewicz96}, and Grädel and Walukiewicz' extensions~\cite{GW06}), but also of Zielonka~\cite{Zielonka98} (in a completely different way) for parity games, and was also formalized by Klarlund~\cite{Klarlund91, Klarlund92} in his notion of progress measures for Rabin games.

Our first contribution is rather conceptual and consists in streamlining the argument, and in particular expliciting the measuring structure as a (well-monotone) graph.
We believe that this has two advantages.
\begin{itemize}
\item[$(i)$] Separating the strategy-folding argument from the universality argument improves conceptual clarity.
\item[$(ii)$] Perhaps more importantly, well-monotone graphs then appear as concrete and manageable witnesses for positionality.
\end{itemize}
We supplement $(ii)$ with our main technical and conceptual novelty in the form of a converse: any positional valuation which has a neutral color admits universal well-monotone graphs.
Stated differently, for such valuations, existence of universal well-monotone graphs characterizes positionality.
This is the first known characterization result for positionality (for one player).
Section~\ref{sec:preliminaries} gives necessary definitions, and Section~\ref{sec:main_results} states and proves our main characterization result.

\vskip1em

We then proceed in Section~\ref{sec:examples} to apply our new framework to a number of well-studied positional objectives and valuations.
In each case, we construct universal well-monotone graphs establishing positionality over arbitrary graphs.
This is an occasion for us to introduce a few different tools for constructing well-monotone universal graphs.
For some of the more advanced valuations and objectives we discuss, our positionality results are novel.

In Section~\ref{sec:lexico}, we study lexicographic products of prefix-independent objectives, which are supported by a natural lexicographic product of well-monotone graphs.
Our main result on this front proves that universality is preserved through lexicographical product.
Combining this with our characterization result, we establish a general closure of prefix-independent positional objectives under finite lexicographical products; this result was not known prior to our work.

Lastly, in Section~\ref{sec:healing} we propose a general class of prefix-independent  positional objectives whose unions are positional.
This class coincides with objectives for which a somewhat naive combination of well-monotone graphs yields universality for the union.
This provides a wide, and quite easy to check, sufficient condition for closure under (countable) union to preserve positionality.

\paragraph*{Comparison with conference version.}
This paper is based on~\cite{Ohlmann22}; however a number of improvements were made.
First, all proofs have been reworked, and some of them considerably simplified.
For example, we removed the need for introducing progress measures in the proof of one of our main results, Theorem~\ref{thm:structure_implies_positionality}, making it conceptually much simpler.

Perhaps more importantly, we introduced some additional content: a study of $K$-monotone objectives and their positionality (Section~\ref{subsec:omega_regular}), stronger positionality result for counter-based valuations (Section~\ref{subsec:quantitative_examples}) and positionality of $\FinParity$ (Section~\ref{subsec:finitary_parity}).
Our results on closure under union (Section~\ref{sec:healing}) are also new to this final version.

\section{Preliminaries}\label{sec:preliminaries}

We use $\powerset X$ to denote the set of subsets of a set $X$, and $\nepowerset X$ for the set of nonempty subsets of $X$.
Throughout the paper, lexicographical products are denoted little-endian style, meaning that the first coordinate is always the weakest.

\paragraph{Graphs.}  In this paper, graphs are directed, edge-colored and have no sinks.
Formally, given a set of colors $C$, a $C$-pregraph $G$ is given by a set of vertices $V(G)$ and a set of edges $E(G) \subseteq V(G) \times C \times V(G)$.
Note that no assumption is made in general regarding the finiteness of $C$, $V(G)$ or $E(G)$.
For convenience, we write $v \re c v'$ for the edge $(v,c,v')$.
If $v \re c v' \in E$, we say that $v$ is a $c$-predecessor of $v'$, and $v'$ is a $c$-successor of $v$.

A sink in a $C$-pregraph is a vertex with no successor.
A $C$-graph is a $C$-pregraph with no sink.
It is often convenient to write $v \re c v' \tin G$, or even simply $v \re c v'$ if $G$ is clear from context, instead of $v \re c v' \in E(G)$.
It is also often the case that $C$ is fixed and clear from context, so we generally just say ``graph'' instead of ``$C$-graph''.
The size (or cardinality) of a graph is defined to be $|V(G)|$.

A path $\pi$ in $G$ is a finite or infinite sequence of edges whose endpoints match, formally 
\[
\pi = (v_0 \re {c_0} v_1)(v_1 \re {c_1} v_2) \dots,
\]
which for convenience we denote by
\[
\pi=v_0 \re{c_0} v_1 \re{c_1} v_2 \dots.
\]
 We say that $\pi$ starts in $v_0$ or that it is a path from $v_0$.
By convention, the empty path $\eps$ starts in all vertices.
The coloration of $\pi$ is the (finite or infinite) sequence $\col(\pi)=c_0c_1 \dots$ of colors appearing on edges of $\pi$.

A non-empty finite path of length $i>0$ is of the form $\pi=v_0 \re{c_0} \dots \re{c_{i-1}} v_i$ and we say in this case that $\pi$ is a path from $v_0$ to $v_i$, and that $v_i$ is the last vertex of $\pi$.
We write
\[
\pi : v \rp w v' \tin G
\]
to say that $\pi$ is a finite path from $v$ to $v'$ with coloration $w \in C^*$ in the graph $G$.
We also write
\[
\pi : v \rp w \tin G
\]
to say that $\pi$ is an infinite path from $v$ with coloration $w \in C^\omega$ in $G$.


Given two graphs $G$ and $G'$, a morphism $\phi$ from $G$ to $G'$ is a map $\phi:V(G) \to V(G')$ such that for all $v \re c v' \in E(G)$, it holds that $\phi(v) \re c \phi(v') \in E(G')$.
Note that $\phi$ need not be injective.
We write $G \re \phi G'$ when $\phi$ is a morphism from $G$ to $G'$, and $G \re{} G'$ when there exists a morphism from $G$ to $G'$.
A subgraph of $G$ is a graph $G'$ such that $V(G') \subseteq V(G)$ and $E'(G) \subseteq E(G)$ (equivalently, the inclusion $V'(G') \to V(G)$ defines a morphism); note that it is assumed for $G'$ to be a graph, it is therefore without sinks.
Given a graph $G$ and a vertex $v_0$ in $G$, we let $G[v_0]$ denote the restriction of $G$ to vertices reachable from $v_0$ in $G$ (note that it is indeed sinkless).

\paragraph*{Games.}

We fix a set of colors $C$.
A $C$-valuation is a map $\val : C^\omega \to X$ from infinite colorations to a set of values $X$ equipped with a complete linear order $\leq$ (a linear order which is also a complete lattice, that is, admits arbitrary suprema and infima).
Given a $C$-graph $G$ and a vertex $v \in V(G)$, the $\val$-value of $v$ is the supremum value of a coloration from $v$:
\[
    \val_G(v) = \sup_{v \rp w \tin G} \val(w).
\]

A $C$-game is a tuple $\game = (G,\VE,\val)$, where
$G$ is a $C$-graph, $\VE \subseteq V(G)$ is the set of vertices controlled by the protagonist, and $\val$ is a $C$-valuation.
To help intuition, we call the protagonist Eve, and the antagonist Adam; Eve seeks to minimize the valuation whereas Adam seeks to maximize it.
We let $\VA$ denote the complement of $\VE$ in $V(G)$.
We now fix a game $\game=(G,\VE,\val)$.

It is convenient in this work to formalize strategies by using graphs and morphisms; we give a formal definition which is explained below.
A strategy from $v_0$ in $\game$ is a tuple $\cS = (S,\pi,s_0)$ consisting of a graph $S$ together with a morphism $\pi : S \to G$ and an initial vertex $s_0 \in V(S)$ such that $\pi(s_0) = v_0$, satisfying that
\begin{center}    
for all edges $v \re c v'$ in $G$ with $v \in \VA$, and for all $s \in \pi^{-1}(v)$,\\ there is an edge $s \re c s'$ in $S$ with $s' \in \pi^{-1}(v')$.
\end{center}
This is illustrated in Figure~\ref{fig:strategy}.

\begin{figure}[h]
\begin{center}
\includegraphics[width=0.5 \linewidth]{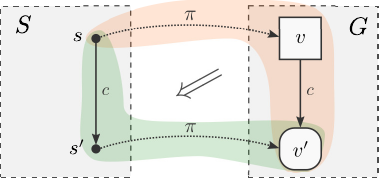}
\end{center}
\caption{Diagram illustrating the definition of a strategy. We use squares to represent vertices controlled by Adam.}\label{fig:strategy}
\end{figure}

Intuitively, a strategy $\S=(S,\pi,s_0)$ from $v_0 \in G$ is used by Eve to play in the game $\game$ as follows:
\begin{itemize}
\item whenever the game is in a position $v \in V(G)$, the strategy is in a position $s \in \pi^{-1}(v)$;
\item at each step, positions in both the game and the strategy are updated along an edge of the same color;
\item initially, the position in the game is $v_0$, and the position in the strategy is $s_0 \in \pi^{-1}(v_0)$;
\item if the position $v$ in the game belongs to $\VA$, and Adam chooses the edge $v \re c v'$ in $G$, then Eve follows an edge $s \re c s'$ in $S$ with $\pi^{-1}(s')$, which exists by definition of $\S$;
\item if the position $v$ in the game belongs to $\VE$, then Eve chooses an outgoing edge $s \re c s'$ in~$S$, which exists since $S$ is sinkless, and progresses in $G$ along the edge $\pi(s) \re c \pi(s') =\linebreak[3] v \re c v' \in E(G)$.
\end{itemize}
Note that infinite paths from $v_0$ in $G$ which are visited when playing as above are exactly infinite paths from $s_0$ in $S$.

We let $\Sigma^\game_{v_0}$ denote the set of strategies from $v_0$ in $\game$.
The value of a strategy $\cS = (S,\pi,v_0)$ is defined to be the value of $s_0$ in $S$, that is:
\[
\val(\cS)=\val_S(s_0)= \sup_{s_0 \rp w \tin S} \val(w).
\]
The value of a vertex $v_0 \in V(G)$ in the game $\game$ is the infimum value of a strategy from $v_0$:
\[
\val_\game(v_0) = \inf_{\cS \in \Sigma_{v_0}^\game} \val(\cS).
\]
A strategy $\cS$ from $v_0$ in $\game$ is called optimal if $\val(\cS)=\val_\game(v_0)$.
Note that there need not exist optimal strategies, as it may be that the value is reached only in the limit.
Note also that we always take the point of view of Eve, the minimizer.
In particular, we will make no assumption on the determinacy of the valuation; in this work, strategies for Adam are irrelevant.

A positional strategy is a strategy which makes choices depending only on the current vertex, regardless of how it was reached.
Formally, a strategy $\cP = (P,\pi,p_0)$ is positional if $V(P) \subseteq V(G)$ and $\pi$ is the inclusion morphism (in particular, $p_0=v_0$); stated differently $P$ is a subgraph of $G$.
Since $\pi$ and $p_0$ carry no information, we will simply say with a slight abuse that the subgraph $P$ of $G$ is a positional strategy.

A positional strategy $P$ with $V(P) = V(G)$ is optimal if for all vertices $v_0 \in V(G)$, it holds that $\val_\game(v_0)=\val_P(v_0)$.
A valuation $\val$ is said to be positional if all games with valuation $\val$ admit an optimal positional strategy.

Two remarks are in order.
First, note that we require positionality over arbitrary (possibly infinite) game graphs.
Second, the concept we discuss is that of uniform positionality, meaning that the positional strategy should achieve an optimal value from any starting vertex.

\paragraph*{Neutral colors.}
Consider a $C$-valuation $\val$.
A color $\eps \in C$ is said to be neutral if
\begin{itemize}
\item for all $w=w_0 w_1 \dots \in C^\omega$ and all $w' = \eps^{n_0} w_0 \eps^{n_1} w_1 \eps^{n_2} \dots$, where $n_0,n_1, \dots \in \N$, it holds that
\[
\val(w) =  \val(w'); \text{ and}
\]
\item for all $u \in C^*$, it holds that
\[
    \val(u \eps^\omega) = \min_{v \in C^\omega} \val(uv).
\]
\end{itemize}
Let $C$ be a set of colors and $\eps \notin C$; we use $C^\eps$ to denote $C \cup \{\eps\}$.
For any $C$-valuation $\val$, there is a unique extension $\val^\eps$ of $\val$ to a $C^\eps$-valuation for which $\eps$ is a neutral color.


\section{Main characterization result}\label{sec:main_results}

In this section, we introduce required definitions and state our main result (Section~\ref{subsec:definitions}), then we prove both directions (Sections~\ref{subsec:structure_implies_positionality} and~\ref{subsec:positionality_implies_structure}).
We also explain how our definitions can be simplified in the important class of prefix-independent objectives (Section~\ref{subsec:prefix_increasing_objectives}).

\subsection{Universality, monotonicity, and statement of the characterization}\label{subsec:definitions}

We now introduce the two main concepts for our characterization of positionality, namely universality and monotonicity.
We fix a set of colors $C$.

\paragraph*{Universal graphs.}

Fix a valuation $\val:C^\omega \to X$.
Recall that values of vertices in a graph~$G$ are given by
\[
\val_G(v)=\sup_{v \rp w \tin G} \val(w).
\]
Given two graphs $G$ and $G'$ with a morphism $\phi:G \to G'$, since there are more colorations from $\phi(v)$ in $G'$ than from $v \in G$ we have in general
\[
\val_G(v) \leq \val_{G'}(\phi(v)).
\]
We say that $\phi$ is $\val$-preserving if the converse inequality holds: for all $v \in V$, $\val_{G}(v)=\val_{G'}(\phi(v))$.

Given a cardinal $\kappa$, we say that a graph $G$ is $(\kappa,\val)$-universal if every graph $H$ of cardinality $< \kappa$ has a $\val$-preserving morphism towards $G$.
We say that a graph is uniformly $\val$-universal if it is $(\kappa,\val)$-universal for all cardinals $\kappa$.
Intuitively, a universal graph should be rich enough to embed all small graphs, but it should do so without introducing paths of large valuations.

For an instructive non-example, consider the graph $G$ with a single vertex and $c$-loops around it for all colors $c$.
It embeds all graphs, however all colorations can be realised as paths in~$G$, so the embedings are not $\val$-preserving (unless $\val$ is constant).
Many examples of universal graphs for various valuations will be discussed in Section~\ref{sec:examples}.

\paragraph*{Monotone graphs.}

A $C$-graph $G$ is monotone if $V(G)$ is equipped with a linear order\footnote{For convenience, we write linear orders as $\geq$ by default (instead of the generally preferred $\leq$). This choice visually aligns better with the composition with edges.}~$\geq$ which is well-behaved with respect to the edge relations, in the sense that for any $u,u',v,v' \in V(G)$ and $c \in C$,
\[
u \geq v \re c v' \geq u' \tin G \quad \implies \quad u \re c u' \tin G.
\]
We call this property monotone composition (with respect to $\geq$) in $G$.
An example is given in Figure~\ref{fig:monotone_graph}.

\begin{figure}[ht]
    \begin{center}
    \includegraphics[width=\linewidth]{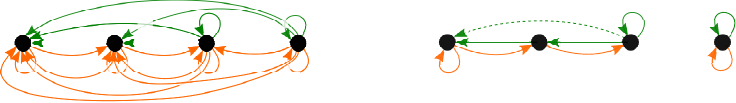}
    \end{center}
    \caption{On the left, a monotone graph, with the leftmost vertex corresponding to the smallest (we will always display ordered graphs using this convention). On the right, only the edges corresponding to min-predecessors are depicted. All other edges (such as the dashed one) can be recovered by composition.}\label{fig:monotone_graph}
\end{figure}

A well-monotone graph $G$ is a monotone graph which is well-ordered (for an order satisfying the monotone composition).

\paragraph*{Statement of the main result.}

We are now ready to state our main characterization result.

\begin{theorem}\label{thm:main}
Let $\val$ be a valuation.
If for all cardinals $\kappa$, there exists a $(\kappa,\val)$-universal well-monotone graph, then $\val$ is positional.
The converse also holds assuming $\val$ has a neutral color.
\end{theorem}

We do not know whether the characterization is complete, meaning, if the assumption on the neutral color can be lifted.
Actually, it can be seen that this is equivalent to asking whether there exists a positional\footnote{Here, our working assumption that positionality refers to uniform positionality is important: if one chooses $\val$ to be positional but not uniformly, then $\val^\eps$ cannot be positional.} valuation $\val$ such that $\val^\eps$ is not positional.
We conjecture that for all positional valuations $\val$, it does hold that $\val^\eps$ is also positional (and therefore our characterization is complete).

Sections~\ref{subsec:structure_implies_positionality} and~\ref{subsec:positionality_implies_structure} prove both directions in Theorem~\ref{thm:main}.

\subsection{Structure implies Positionality}\label{subsec:structure_implies_positionality}

We prove the first implication in Theorem~\ref{thm:main}, stated as follows.

\begin{theorem}\label{thm:structure_implies_positionality}
Let $\val$ be a $C$-valuation such that for all cardinals $\kappa$ there exists a $(\kappa,\val)$-universal well-monotone graph.
Then $\val$ is positional.
\end{theorem}

Our proof is inspired by those of Emerson and Jutla~\cite{EJ91} (see also the presentation by Walukiewicz~\cite{Walukiewicz96}) and Klarlund~\cite{Klarlund92}, respectively for parity games and Rabin games.
In the conference version of this paper~\cite{Ohlmann22}, we gave a proof using progress measures; here however we refrain from introducing them and give a more direct presentation of essentially the same argument.

Let us first give a high-level overview of the proof.
Our aim is to ``fold'' a given strategy $\S$ into a positional one achieving a better (that is, a smaller) value; the challenge is that different occurrences of a vertex $v \in \VE$ in the strategy $\S$ have different outgoing edges.
We overcome this challenge by picking a well-monotone graph $U$ into which the strategy has a $\val$-preserving morphism $\phi$.
The crucial observation is that by well-foundedness of $U$, (at least) one of the occurrences of $v$ in $\S$ is mapped by $\phi$ to a minimal position.
We then define the positional strategy~$P$ by mimicking outgoing edges from this special occurrence of $v$ in the strategy.
Monotonicity of $U$ will then guarantee that the values in $P$ are small, as required.
In fact, since we aim for a single optimal positional strategy $P$, we will apply this argument to $\S$ being a disjoint union of many strategies, whose values converge to the optimal values.

We will make use of the fact that in a monotone graph $U$, the value (in fact, the set of colorations) increases with the order: for all $u,u' \in V(U)$,
\[
    u \geq u' \quad \implies \quad \val_U(u) \geq \val_U(u').
\]

Indeed, by monotone composition, for any path $u' \re{c_0} u_1 \re{c_1} \dots$ from $u'$ in $U$ it holds that $u \re{c_0} u_1 \re{c_1} \dots$ is a path from $u$ in $U$ with the same coloration, implying the result.

We now fix a game $\game=(G,\VE, \val)$ with valuation $\val$.
For each $v \in V(G)$, we fix a sequence $\S_v^{(k)} = (S_v^{(k)}, \pi_v^{(k)}, s_{0,v}^k)$ of strategies from $v$ in $\game$ such that $\val(\S_v^{(k)}) \re{k \to \infty} \val_\game(v)$.
We then let the graph $S$ be the disjoint union of all the $S_v^{(k)}$'s, where $k$ ranges over $\N$ and $v$ over $V(G)$.
We let $\pi : V(S) \to V(G)$ be the unique extension of all the $\pi_v^{(k)}$; it defines a morphism $S \to G$.

We now pick a cardinal $\kappa > |S|$, let $U$ be a $(\kappa,\val)$-universal well-monotone graph, and fix a $\val$-preserving morphism $\phi: S \to U$.
Consider the map $\phi': V(G) \to V(U)$ defined by
\[
    \phi'(v) = \min \phi(\pi^{-1}(v)),
\]
which exists since $U$ is well-ordered; see Figure~\ref{fig:structure_to_positionality} for an illustration.
We define a subgraph $P$ of~$G$ with $V(P) = V(G)$ by
\[
    v \re c v' \tin P \quad \iff \quad \exists s \in \pi^{-1}(v), \phi(s) = \phi'(v), s \re c s' \in \pi^{-1}(v') \tin S.
\]
Intuitively, to define edges in $P$ outgoing from $v \in V(G)$, we look at strategy vertices $s$ in $\pi^{-1}(v)$ which are evaluated in $U$ to be optimal (meaning that $\phi(s)$ is minimal).

\begin{figure}[h]
\begin{center}
\includegraphics[width= \linewidth]{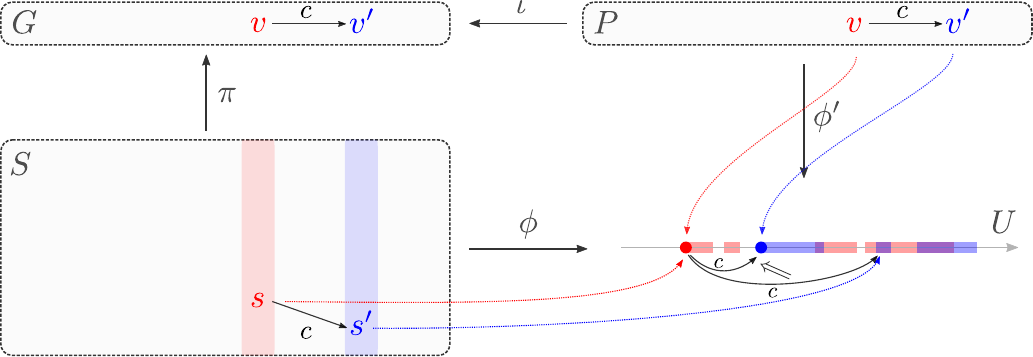}
\end{center}
\caption{An illustration supporting the reader in the proof of Theorem~\ref{thm:structure_implies_positionality}.}\label{fig:structure_to_positionality}
\end{figure}

We make two claims: first, $P$ is a strategy in $G$ (hence it is a positional strategy), second, $\phi'$ defines a morphism from $P$ to $U$.
Before proving these two claims, we argue that this implies the wanted result; for this we should prove that $P$ is optimal, meaning that for all $v \in V(G)$ we have $\val_P(v) = \val_\game(v)$.

Let $v \in V(G)$.
Since $\phi'$ is a morphism,
\[\label{eq:1}\tag{1}
    \val_P(v) = \inf_{v \rp w \tin P} \val(w) \leq \inf_{\phi'(v) \rp w \tin U} \val(w) = \val_U(\phi'(v)).
\]
Now for all $k \in \N$, we have in $U$ that
\[
    \phi'(v) = \min_{s \in \pi^{-1}(v)} \phi(s) \leq \phi(s_{0,v}^{(k)}),
\]
and thus, by the remark above ($\val$ increases with the order over $U$), and the fact that $\phi$ is $\val$-preserving, we get
\[\label{eq:2}\tag{2}
    \val_U(\phi'(v)) \leq \val_U(\phi(s_{0,v}^{(k)})) = \val_{S_v^{(k)}}(s_{0,v}^{(k)}) = \val(\S_v^{(k)})
\]
Combining~(\ref{eq:1}) and~(\ref{eq:2}) we obtain
\[
    \val_P(v) \leq \inf_k \val(\S_v^{(k)}) \leq \val_\game(v),
\]
so $P$ is indeed optimal.

It remains to prove our two claims, we start by verifying that $P$ is a strategy.
Let $v \in \VA$, and $v \re c v' \tin G$, we must prove that $v \re c v' \tin P$.
Let $s \in \pi^{-1}(v)$ be such that $\phi(s) = \phi'(v)$.
Let $k \in \N$ and $v'' \in V(G)$ be such that $s \in S_{v''}^{(k)}$.
Since $S_{v''}^{(k)}$ is a strategy and $\pi_{v''}^{(k)}(s) = \pi(s) = v \in \VA$ there is $s' \in S_{v''}^{(k)}$ such that $\pi(s')=v'$ and $s \re c s' \tin S_{v''}^{(k)}$ and thus $s \re c s' \tin S$.
By definition of $P$, this proves that $v \re c v' \tin P$, as required.

We now prove that $\phi':P \to U$ is a morphism.
Let $v \re c v' \tin P$.
Let $s \in \pi^{-1}(v)$ be such that $\phi(s) = \phi'(v)$, and $s' \in \pi^{-1}(v')$ such that $s \re c s' \tin S$.
Since $\phi: S \to U$ is a morphism, we have $\phi(s) \re c \phi(s') \tin U$.
Hence
\[
    \phi'(v) = \phi(s) \re c \phi(s') \geq \phi'(v') \tin U, 
\]
and thus by monotone composition, $\phi'(v) \re c \phi'(v') \tin U$, the wanted result.

\subsection{Positionality implies Structure}\label{subsec:positionality_implies_structure}

We now establish our main technical novelty which is stated as follows.

\begin{theorem}\label{thm:structuration}
Let $\val$ be a positional $C$-valuation admitting a neutral color, and let $G$ be a $C$-graph.
There exists a well-monotone $C$-graph $G'$ with a $\val$-preserving morphism $G \to G'$.
\end{theorem}

We obtain the converse implication in Theorem~\ref{thm:main} as a consequence.

\begin{corollary}\label{cor:positionality_implies_structure}
Let $\val$ be a positional $C$-valuation admitting a neutral color.
For all cardinals~$\kappa$, there exists a well-monotone $(\kappa,\val)$-universal graph.
\end{corollary}

\begin{proof}[Proof of Corollary~\ref{cor:positionality_implies_structure}]
Let $\kappa$ be a cardinal, and let $G$ be the disjoint union of all $C$-graphs of cardinality $< \kappa$, up to isomorphism.
Note that $G$ it is $(\kappa,\val)$-universal.
Theorem~\ref{thm:structuration} gives a well-monotone graph $G'$ which has a $\val$-preserving morphism $G \to G'$; now $G'$ is $(\kappa,\val)$-universal by composition of $\val$-preserving morphisms.
\end{proof}

\paragraph{Proof overview.}

We now fix a positional $C$-valuation $\val$, with a neutral color~$\eps$, and a graph~$G$.
Our proof consists of the two following steps:
\begin{itemize}
\item[$(i)$] add many $\eps$-edges to $G$ while preserving $\val$; then
\item[$(ii)$] add even more edges by closing around $\eps$-edges (this is made formal below), and quotient by $\re \eps$-equivalence.
\end{itemize}
For the second step to produce a well-monotone graph, we need to guarantee that there are sufficiently many $\eps$-edges which were added in the first step.
We start by the second step; in particular, we formalize what ``sufficiently many'' means.
The first step is more involved and exploits positionality of $\val$.

\paragraph{Second step.}

We say that a graph $G$ has sufficiently many $\eps$-edges if $\re \eps$ is well-founded, that is,
\[
\forall A \in \nepowerset {V(G)}, \exists v \in A, \forall v' \in A, v' \re \eps v \tin G.
\]
The statement below reduces our goal to that of adding many $\eps$-edges to $G$.

\begin{lemma}\label{lem:second_step}
If $G$ has sufficiently many $\eps$-edges then there exists a well-monotone graph $G'$ with a $\val$-preserving morphism $G \to G'$.
\end{lemma}

\begin{proof}
We first define the $\eps$-closure $G_1$ of $G$ by $V(G_1) = V(G)$ and
\[
v \re c v' \tin G_1 \quad \iff \quad \exists u,u' \in V(G), v \rp{\eps^*} u \re c u' \rp{\eps^*} v' \tin G,
\]
where $x \rp{\eps^*} y$ is a shorthand for ``there exists $n \in \N$ such that $x \rp {\eps^n} y$''.
Note that we have $E(G) \subseteq E(G_1)$; we claim that the identity morphism from $G$ to $G_1$ is in fact $\val$-preserving.
Indeed, for all $v_0 \in V(G)$ and for all infinite paths $\pi': v_0 \rp {w'}$ in $G_1$, there is a path $\pi: v_0 \rp w$ in $G$ with $w = \eps^{n_0} w'_0 \eps^{n_1} w'_1 \dots$.
By neutrality of $\eps$ we have $\val(w') = \val(w)$, and thus
\[
\val_{G_1}(v_0) = \sup_{v_0 \rp w' \tin G_1} \val(w') \leq \sup_{v_0 \rp w \tin G} \val(w) = \val_{G}(v_0).
\]

Note that in $G_1$, $\re \eps$ satisfies monotone composition with all colors, and in particular it is transitive (by taking $c=\eps$).
It is moreover well-founded (and thus also total) and reflexive thanks to the assumption of the Lemma.
However, it is not antisymmetric, which is why we now quotient with respect to $\re \eps$-equivalence.

Formally, we define $\sim$ over $V(G)$ by 
\[
v \sim v' \iff (v \re \eps v' \tand v' \re \eps v) \tin G_1,
\]
which is an equivalence relation.
Note that vertices which are $\sim$-equivalent have the same incoming and outgoing edges in $G_1$ (since $G_1$ is $\eps$-closed), therefore the graph $G_2$ with $V(G_2) = V(G) / \sim$ given by
\[
{[v]}_\sim \re c {[v']}_\sim \tin G_2 \iff v \re {c} v' \tin G_1,
\]
is well-defined, and moreover colorations from ${[v]}_\sim$ in $G_2$ are the same as colorations from~$v$ in~$G_1$.
Hence the projection $V(G) \to V(G) / \sim$ defines a $\val$-preserving morphism from~$G_1$ to~$G_2$, which, composed with the identity morphism $G \to G_1$, gives a $\val$-preserving morphism $G \to G_2$.
Now $\re \eps$ is a well-founded order satisfying monotone composition in~$G_2$, therefore $G_2$ is well-monotone.
\end{proof}

Note that the second step has not made use of positionality of $\val$; it is exploited below.

\paragraph{First step.}
We now show that sufficiently many edges can be added to $G$ while preserving $\val$, thanks to its positionality.
We consider the game $\game'=(G',\VE',\val)$ given by $V(G') = V(G) \cup \nepowerset {V(G)}$, $\VE' = \nepowerset {V(G)}$, and
\[
E(G') =  E(G) \cup \{v \re {\eps} A \mid v \in A\} \cup \{A \re {\eps} v \mid v \in A\}.
\]
An example is given in Figure~\ref{fig:multiple_choice}.

\begin{figure}[ht]
\begin{center}
\includegraphics[width=0.7\linewidth]{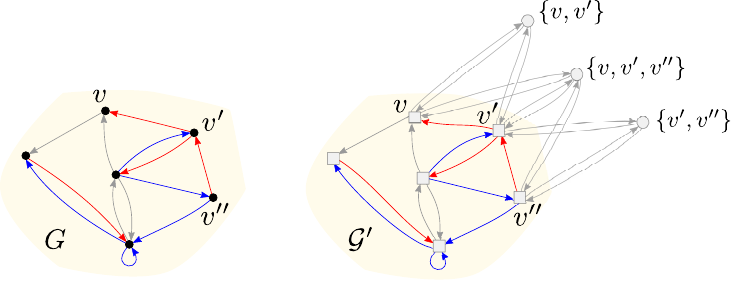}
\end{center}
\caption{On the left, a $\{\tred, \tblue, \tgray\}$-graph $G$ with $\eps=\tgray$.
On the right, the corresponding game $\game'$, where only $3$ of the $2^6 -1$ Eve-vertices (circles) are represented for clarity.}\label{fig:multiple_choice}
\end{figure}

When playing in $\game'$, Adam follows a path in $G$, with the additional possibility, at any point, to switch from a vertex $v$ to an Eve-vertex $A$ containing $v$.
It is then left to Eve to choose a successor from $A$, which can be any vertex in $A$.
A natural choice is then to go back to $v$, which guarantees a small value thanks to neutrality of $\eps$.

\begin{lemma}\label{lem:G_is_easy}
For all $v_0 \in V(G)$, we have $\val_{\game'}(v_0) \leq \val_G(v_0)$.
\end{lemma}

Note that the converse inequality also holds (it is however irrelevant for us).

\begin{proof}
Let $v_0 \in V(G)$.
We consider the strategy $\S=(S,\pi,s_0)$ from $v_0$ over $\game'$ described above.
Formally, we let $V(S) = V(G) \cup \nepowerset{V(G)} \times V(G)$ and
\[
    E(S) = E(G) \cup \{v \re \eps (A,v) \mid v \in A\} \cup \{(A,v) \re \eps v \mid v \in A\},
\]
for $v \in V(G)$, $\pi(v) = v$ and for $(A,v) \in \nepowerset{V(G)} \times V(G)$, $\pi(A,v)=A$, and $s_0=v_0$.
It is a direct check that $\S$ is indeed a strategy over $\game'$: $\pi$ is indeed a morphism, $S$ is indeed a graph, and for $v \in \VA' = V(G)$, $\pi^{-1}(v)=\{v\} \subseteq V(S)$, and all edges outgoing from $v$ in $G'$ there is a corresponding edge in $S$.

Now observe that infinite paths from $s_0=v_0$ in $S$ either have colorations of the form $w'=\eps^{n_0}w_0\eps^{n_1}w_1 \dots$, where $w=w_0w_1 \dots$ is a coloration of an infinite path from $v_0$ in $G$, or of the form $w'=\eps^{n_0}w_0 \dots w_i\eps^\omega$, where $w=w_0 \dots w_i$ is a coloration of a finite path from $v_0$ in $G$.
Thus thanks to neutrality of $\eps$,
\[
\begin{array}{lcl}
\val_{\game'}(v_0) \leq \val(\S) &=& \sup_{s_0 \rp {w'} \tin S} \val(w')\\
& \leq & \sup_{v_0 \rp w \tin G} \val(w) = \val_G(v_0),
\end{array}
\]
concluding the proof.
\end{proof}

Note that the above strategy $\S$ is somewhat far from being positional: each vertex $A \in \VE'$ requires memory $\pi^{-1}(A) = \{A\} \times V(G)$.
However, thanks to positionality of $\val$, which we now exploit, there exist positional strategies achieving the same value as $\S$.
Observe that a positional strategy $P'$ in $\game'$ corresponds to the choice of (at least) a successor $v' \in A$ for each nonempty subset $A$ of $V(G)$.

\begin{lemma}\label{lem:adding_edges}
Let $P'$ be an optimal positional strategy in $\game'$, and let $G''$ be the graph defined by $V(G'')=V(G)$ and
\[
E(G'') = E(G) \cup \{v \re \eps v' \mid \exists A \in \nepowerset{V(G)}, v,v' \in A \tand A \re \eps v' \tin P'\}.
\]
The identity defines a $\val$-preserving morphism from $G$ to $G''$.
Moreover, $G''$ has sufficiently many $\eps$-edges.
\end{lemma}

\begin{proof}
By optimality of $P'$, and thanks to Lemma~\ref{lem:G_is_easy}, we have for all $v_0 \in V(G)$ that $\val_{P'}(v_0) \leq \val_G(v_0)$.
Consider an infinite path $\pi''$ in $G''$;
it is of the form
\[
\pi'' : v_0 \rp {w_0} v'_0 \re \eps v_1 \rp {w_1} v_1' \re \eps v_2 \dots,
\]
where for each $i$, $v_i \rp {w_i} v'_i$ is a path in $G$ and there exists $A_i \in \nepowerset{V(G)}$ such that $v'_i,v_{i+1} \in A_i$ and $A_i \re \eps v_{i+1} \tin P'$.
Therefore, there is a path of the form
\[
\pi' : v_0 \rp {w_0} v'_0 \re \eps A_0 \re \eps v_1 \rp {w_1} v_1' \re \eps A_1 \re \eps v_2 \dots
\]
in $P'$ (the edges $v'_i \re \eps A_i$ as well as the paths $v_i \rp {w_i} v_i'$ belong to $P'$ since it is a strategy).
By neutrality of $\eps$ we have $\val(\col(\pi''))=\val(\col(\pi'))$ and hence
\[
\begin{array}{rcl}
\val_{G''}(v_0) & = & \sup_{v_0 \rp {w''} \tin G''} \val(w'') \\
& \leq & \sup_{v_0 \rp {w'} \tin P'} \val(w') \\
& = & \val_{P'}(v_0) \leq \val_G(v_0),
\end{array}
\]
thus the identity is $\val$-preserving from $G$ to $G''$.

Finally, each non-empty $A \subseteq V(G)$ has an $\eps$-successor $v'$ in $P'$, which satisfies that each $v \in A$ has an $\eps$-edge towards $v'$ in $G''$.
Stated differently, $G''$ has sufficiently many $\eps$-edges.
\end{proof}

We conclude with Theorem~\ref{thm:structuration} by combining Lemmas~\ref{lem:adding_edges} and~\ref{lem:second_step}.

\subsection{Specialization to prefix-independent objectives}\label{subsec:prefix_increasing_objectives}

In this section, we show that our notion of universality instantiates to that introduced by Colcombet and Fijalkow~\cite{CF18} (and studied over finite graphs and prefix-independent objectives), in the case of prefix-increasing objectives.
This allows to simplify the definitions in the very commonly studied case of prefix-independent objectives.
We also discuss interactions between universality and pregraphs; in particular we show that for prefix-independent objectives, nonempty universal graphs also embed pregraphs.

\paragraph*{Prefix-independence.} We call a $C$-valuation $\val$ prefix-increasing (resp.~prefix-decreasing) if adding a prefix increases (resp. decreases) the value: for any $w \in C^\omega$ and $u \in C^*$, we have $\val(uw) \geq \val(w)$ (resp. $\val(uw) \leq \val(w)$).
If a valuation is both prefix-increasing and prefix-decreasing, we say that it is prefix-independent.
We have the following technical lemma.

\begin{lemma}\label{lem:no_val_increasing_edge}
Assume that $\val$ is prefix-increasing and consider a graph~$G$.
If two vertices~$v$ and~$v'$ satisfy $\val_G(v) < \val_G(v')$ then there is no edge in $G$ from $v$ to $v'$.
\end{lemma}

\begin{proof}
By contradiction, let $e= v \re c v'$ be an edge in $G$ and pick a path $\pi'$ from $v'$ with $\val_G(\pi') > \val_G(v)$.
Then $e \pi'$ is a path from $v$ and we have
\[
\val(e\pi') \geq \val(\pi') > \val(v),
\]
which is a contradiction since $e\pi'$ is a path from $v$ in $G$.
\end{proof}

\paragraph*{Objectives.} We say that $\val: C^\omega \to X$ is an objective if $X$ is the ordered pair $\{\bot,\top\}$.
In this case, we also say that $\val$ is qualitative.
From the point of view of Eve, $\bot$ is interpreted as winning, whereas $\top$ is losing.
Following the usual convention, we identify a qualitative valuation $\val$ with the set $W=\val^{-1}(\bot)$ of infinite words which are winning for Eve.
We say that a vertex $v$ (in a graph) satisfies $W$ if all colorations from $v$ belong to $W$; this amounts to saying that $v$ has value $\bot$.
We also say that a graph satisfies $W$ if all its vertices satisfy $W$; in this case we write $G \models W$.
Note that in the qualitative case, a morphism $G \to G'$ is $W$-preserving if and only if any vertex satisfying $W$ in $G$ is mapped to a vertex satisfying $W$ in $G'$.

\paragraph*{CF-universality.}
Adapting the definition of Colcombet and Fijalkow~\cite{CF18} to infinite cardinals, we say that a graph $G$ is $(\kappa,W)$-CF-universal if

\begin{itemize}
    \item $G \models W$; and
    \item all graphs $H$ such that $|H| < \kappa$ and $H \models W$ have a morphism towards $G$.
\end{itemize}

Given a monotone graph $G$, we let $G^\top$ be the monotone graph obtained by adding a new maximal vertex $\top$ with all possible outgoing edges; formally $V(G^\top) = V(G) \sqcup \{\top\}$ and
\[
    E(G^\top) = E \cup \{\top\} \times C \times V(G^\top).
\]
We now relate the two notions of universality in the special case of a prefix-increasing objectives~$W$, that is, satisfying for all colors~$c$ that $cW \subseteq W$.

\begin{lemma}\label{lem:CF-universality}
Let $W \subseteq C^\omega$ be a prefix-increasing objective and $\kappa$ a cardinal.
\begin{itemize}
    \item If $G$ is a $(\kappa,W)$-universal graph, then its restriction $G_W$ to vertices satisfying $W$ is $(\kappa,W)$-CF-universal. Moreover if $G$ is well-monotone, then so is $G_W$.
    \item If $G$ is a $(\kappa,W)$-CF-universal graph, then $G^\top$ is $(\kappa,W)$-universal. Moreover if $G$ is well-monotone, then so is $G^\top$.
\end{itemize}
\end{lemma}

\begin{proof}
We start with the first item; let $G$ be a $(\kappa,W)$-universal graph and let $G_W$ be its restriction to vertices that satisfy~$W$.
By Lemma~\ref{lem:no_val_increasing_edge}, there is no edge from $V(G_W)$ to its complement in~$G$, therefore $G_W$ is indeed a graph.
It is clear that $G_W$ satisfies $W$.
Now if $H$ is a graph satisfying $W$ and of cardinality $< \kappa$, it has a $W$-preserving morphism $\phi$ into $G$.
The fact that $\phi$ is $W$-preserving means that it is actually a morphism into $G_W$, as required.
Therefore $G_W$ is $(\kappa,W)$-CF-universal.
Finally, assuming $G$ is well-monotone, it is immediate that $G_W$ is well-monotone (as is any restriction of a well-monotone graph).

We now prove the second item; let $G$ be a $(\kappa,W)$-CF-universal graph.
Let $H$ be a graph of cardinality $< \kappa$, and let $H_W$ denote its restriction to vertices satisfying $W$ (it is indeed a graph thanks to Lemma~\ref{lem:no_val_increasing_edge}).
By CF-universality there is a morphism from $H_W$ to $G$, we extend it to a morphism $H \to G^\top$ by mapping vertices not in $H_W$ to $\top$.
It is a morphism thanks to Lemma~\ref{lem:no_val_increasing_edge}, and it is $W$-preserving by definition.
We conclude that $G^\top$ is indeed $(\kappa,W)$-universal.
Finally, it is clear that if $G$ is well-monotone, then so is $G^\top$.
\end{proof}    

In the case of prefix-increasing objectives, since one may translate, thanks to Lemma~\ref{lem:CF-universality}, (well-monotone) CF-universal graphs to (well-monotone) universal graphs and back, we will focus on constructing CF-universal graphs, bipassing the need for systematically introducing an additional $\top$-vertex.
For convenience and by a slight abuse, we will simply say that such graphs are universal.

\paragraph*{Pregraphs and prefix-independent objectives.}

In some applications (see Section~\ref{sec:lexico}), it is more convenient to work with pregraphs (which allow for sinks) rather than graphs.
The definitions remain the same: the valuation of a vertex in a pregraph is the supremum valuation of all infinite colorations from this vertex.
Stated differently, paths ending in sinks are not taken into account, which corresponds to the intuition that they are winning for Eve.

This may be unsatisfactory, for instance if we consider safety games, defined over $C = \{\safe,\bad\}$ by the objective $\Safety = \{\safe^\omega\}$, then we would certainly want finite paths containing occurrences of $\bad$ to be losing rather than winning.
Note that $\Safety$ is not prefix-independent, because $\bad \cdot \Safety \nsubseteq \Safety$.
In contrast, for prefix-independent objectives, working with pregraphs is essentially harmless.

\begin{lemma}\label{lem:pregraphs}
Let $W \subseteq C^\omega$ be a prefix-independent objective, let $\kappa$ be a cardinal number, and let $G$ be a nonempty $(\kappa,W)$-universal well-monotone graph.
Then every pregraph $H$ of size $< \kappa$ that satisfies $W$ has a morphism into $G$.
\end{lemma}

\begin{proof}
Let $v_0 \in V(G)$ be the minimal vertex in $G$, and let $v_0 \re {c_0} v'$ be an edge outgoing from $v_0$ in~$G$.
By monotone composition, the edge $v_0 \re{c_0} v_0$ also belongs to~$G$, and therefore $v_0 \rp {c_0^\omega}$ in~$G$.
Thus since $G \models W$, it must be that $c_0^\omega \in W$.

Now, let $H$ be a pregraph of size $< \kappa$ which satisfies $W$, and let $H'$ be the graph obtained from $H$ by appending a $c_0$-loop to every sink.
Then paths in $H'$ can be of two types: those that belong to $H$, and those that are comprised of a finite prefix followed by infinitely many occurrences of a $c_0$-loop.
Both types of paths satisfy $W$ (since $H \models W$ and by prefix-independence), and thus $H' \models W$.
Thus $H'$ has a morphism into $G$ by universality, which also defines a morphism $H \to G$.
\end{proof}

\section{Examples and first applications}\label{sec:examples}

In this section, we give various examples of constructions of well-monotone graphs which are universal with respect to well-studied conditions.
In particular, thanks to Theorem~\ref{thm:structure_implies_positionality}, this establishes positionality in each case.
We start with some $\omega$-regular objectives, then move on to the study of a few valuations which are inherently quantitative, and finish the section with a study of finitary parity objectives.

\subsection{A few \texorpdfstring{$\omega$}{omega}-regular objectives}\label{subsec:omega_regular}

\paragraph{Safety games.}
The safety objective is defined over $C=\{\safe,\bad\}$ by
\[
\Safety=\{\safe^\omega\}.
\]
It is the simplest in terms of winning strategies: Eve is guaranteed to win as long as she follows a $\safe$-edge which remains in the winning region.
Note that it is prefix-increasing, and thus (see Lemma~\ref{lem:CF-universality}) we are looking for a well-monotone graph satisfying $\Safety$ and which embeds all graphs satisfying $\Safety$.

Now satisfying Safety for a graph simply means not having a $\bad$-edge therefore we have the following result.

\begin{lemma}
The well-monotone graph comprised of a single vertex with a $\safe$-loop is uniformly $\Safety$-universal.
\end{lemma}

This proves thanks to Theorem~\ref{thm:structure_implies_positionality} that safety games are positionally determined (which of course has much simpler proofs).

\paragraph{A variant of Safety.}

For the sake of studying a simple example with no prefix-independence property we consider the objective over $C=\{\imm,\safe,\bad\}$ defined by
\[
W = \imm \cdot \{\imm, \safe\}^\omega.
\]
In words, Eve should immediately see the color $\imm$, and then avoid $\bad$ forever.
Here, $\bad \cdot W \nsubseteq W$ and $W \nsubseteq \safe \cdot W$.
Consider the graph $U$ depicted in Figure~\ref{fig:no_invariance_example}.

\begin{figure}[ht]
\begin{center}
\includegraphics[width=0.4\linewidth]{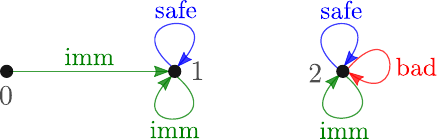}
\end{center}
\caption{A monotone $\{\imm,\safe,\bad\}$-graph $U$ over $V(U)=\{0,1,2\}$.
Edges which follow from composition are not depicted.
Note that neither $1$ nor $2$ satisfy $W$ in $U$.}\label{fig:no_invariance_example}
\end{figure}

\begin{lemma}
The completely well-monotone graph $U$ is uniformly $W$-universal.
\end{lemma}

Therefore $W$ is positionally determined over all graphs.

\begin{proof}
Consider any $C$-graph $G$, and let $V_0,V_1,V_2 \subseteq V(G)$ be the partition of $V(G)$ defined by
\begin{itemize}
\item $v \in V_2$ if and only if $v$ has a path which visits a $\bad$-edge, and
\item $v \in V_0$ if and only if $v \notin V_2$ and all edges outgoing from $v$ have color $\imm$.
\end{itemize}
Note that $V_0$ is precisely the set of vertices which satisfy $W$.
It is immediate that mapping $V_0$ to $0$, $V_1$ to $1$ and $V_2$ to $2$ defines a $W$-preserving morphism from $G$ to $U$.
\end{proof}

\paragraph{Reachability games.}

We now consider the reachability objective over $C=\{\wait,\good\}$, given by
\[
\Reachability = \{w \in C^\omega \mid |w|_{\good} \geq 1\} = C^* \good C^\omega.
\]
Note that $\Reachability$ is not prefix-increasing therefore elements which do not satisfy the objective in the sought monotone graph may play a non-trivial role.
Given an ordinal $\alpha$, we let $U_\alpha$ denote the graph over $V(U_\alpha) = \alpha+1 = [0,\alpha]$ given by
\[
\lambda \re c \lambda' \tin U_\alpha \quad \iff \quad c=\good \tor \lambda > \lambda' \tor \lambda = \alpha.
\]
It is illustrated in Figure~\ref{fig:reachability_construction}.

\begin{figure}[ht]
\begin{center}
\includegraphics[width=0.7\linewidth]{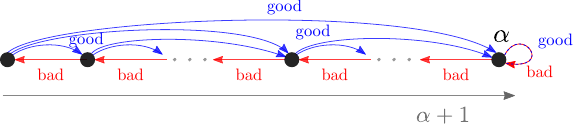}
\end{center}
\caption{The graph $U_\alpha$. Some edges which follow from monotone composition are omitted for clarity (for instance, $\good$-edges pointing from right to left), from now on we no longer mention the use of this convention.
Note that in $U_\alpha$, the vertex $\alpha$ does not satisfy $\Reachability$, however every other vertex does.}\label{fig:reachability_construction}
\end{figure}

At the level of intuition, each path from a vertex satisfying reachability in a given graph eventually visits a $\good$-edge.
There is in fact a well-defined ordinal $\phi(v)$ which captures the number of steps required from $v$ until a $\good$-edge is reached.
This can be rephrased as a universality result.

\begin{lemma}\label{lem:non_uniform_construction_reachability}
For any ordinal $\alpha$, $U_\alpha$ is well-monotone and it is $(|\alpha|,\Reachability)$-universal.
\end{lemma}

The proof provides a template which will later be adapted to other objectives hence we break it into well-distinguished steps.

\begin{proof}
Monotonicity of $U_\alpha$ is a direct check, and it is clear that it is well-founded.
Since there are no infinite paths of $\bad$-edges from vertices $< \lambda$, we have
\[
\lambda \text{ satisfies } \Reachability \tin U_{\alpha} \quad \iff \quad \lambda < \alpha.
\]
We now fix an arbitrary graph $G$.

\begin{itemize}
\item[$(i)$] We construct by transfinite recursion an increasing ordinal-indexed sequence of subsets of $V(G)$ by setting for each ordinal $\lambda$
\[
V_\lambda = \big\{ v \in V(G) \mid v \re c v' \tin G \implies [c = \good \tor \exists \beta < \lambda, v' \in V_\beta]\big\}.
\]
\item[$(ii)$] We let $R = \bigcup_\lambda V_\lambda$ and aim to prove that if $v$ satisfies $\Reachability$ in $G$ then $v \in R$.
We proceed by contrapositive and assume that $v_0 \notin R$: for any ordinal $\lambda$, $v_0 \notin V_\lambda$.
Then $v_0$ has a $\wait$-edge towards some vertex $v_1$ such that for all $\lambda$, $v_1 \notin V_\lambda$.
By a quick induction we build an infinite path $v_0 \re \wait v_1 \re \wait \dots$ in $G$, which guarantees that $v_0$ does not satisfy $\Reachability$.
\item[$(iii)$] We show that if $V_\lambda = V_ {\lambda+1}$ then for all $\lambda' \geq \lambda$ we have $V_{\lambda'} = V_{\lambda}$.
This is direct by transfinite induction: assume the result known for all $\beta$ such that $\lambda \leq \beta < \lambda'$ and let $v \in V_{\lambda'}$.

Then any edge from $v$ is either a $\good$-edge or points towards $v' \in V_\beta$ for some $\beta < \lambda'$, and the result follows since $V_\beta \subseteq V_\lambda$.
\item[$(iv)$] We now let $\alpha$ be such that $|\alpha| > |V(G)|$ and prove that $V_\lambda = V_{\lambda+1}$ for some $\lambda<\alpha$.
Indeed, if this were not the case, then any map (obtained using the axiom of choice)
\[
\begin{array}{lcl}
\alpha & \to & V(G) \\
\lambda &\mapsto& v \in V_{\lambda+1} \setminus V_{\lambda}
\end{array}
\]
would be injective, a contradiction.
\item[$(v)$] Therefore $R=\bigcup_{\lambda < \alpha} V_{\alpha}$ and we let $\phi: V(G) \to V(U_{\alpha}) =[0,\alpha]$ be given by
\[
\phi(v) = \begin{cases}
\min\{\lambda \mid v \in V_\lambda\} & \tif v \in R\\
\alpha &\tif v \notin R.
\end{cases}
\]
By the second item and since $\lambda$ satisfies $\Reachability$ provided it is $<\alpha$, it holds that $\phi$ preserves $\Reachability$. 
\item[$(vi)$] We verify that $\phi$ defines a graph-morphism, which follows from the definitions of $V_\lambda$ and of $U_\alpha$.
First, $\good$-edges are preserved (independently of $\phi$) since they all belong to $U_\alpha$.
Second, $\wait$-edges from $\compl R$ are preserved since $\alpha$ has all outgoing $\wait$-edges in $U_{\alpha}$.
Third if $v \re \wait v'$ is such that $v \in U$ then $\phi(v') < \phi(v)$ by definition of $\phi$ thus $\phi(v) \re \wait \phi(v')$. \qedhere
\end{itemize}
\end{proof}

\paragraph{B\"uchi games.}
The B\"uchi condition is defined over the same set of colors $C=\{\wait,\good\}$ by
\[
\Buchi = \{w \in C^\omega \mid |w|_{\good} = \infty\} = (\wait^* \good)^\omega.
\]
It is prefix-independent so we aim to construct well-monotone graphs which satisfy $\Buchi$ and embed graphs satisfying $\Buchi$.

Given an ordinal $\alpha$, we consider the graph $U_\alpha$ over $V(U_\alpha)=\alpha=[0,\alpha)$ given by
\[
\lambda \re c \lambda' \tin U_\alpha \qquad \iff \qquad c = \good \tor  \lambda > \lambda'.
\]

The difference between the completion $(U_\alpha)^{\top}$ of the graph defined just above for $\Buchi$ and the graph we used for $\Reachability$ is that in the latter there are $\good$-edges towards the maximal element.
This reflects the fact that in a reachability game there may be $\good$-edges from the winning region to its complement, which is of course false in a $\Buchi$-game (precisely because they are prefix-independent).

It is a direct check that $U_\alpha$ is well-monotone and that it satisfies $\Buchi$.
The intuition behind the following result is that one may associate, to any vertex satisfying $\Buchi$ in a given graph, an ordinal corresponding to the number of $\wait$-edges before the next $\good$ edge.

\begin{lemma}
\label{lem:non_uniform_construction_buchi}
For any ordinal $\alpha$, $U_\alpha$ is $(|\alpha|,\Buchi)$-universal.
\end{lemma}

We follow the same steps as those of the proof of Lemma~\ref{lem:non_uniform_construction_reachability}.

\begin{proof}
Fix a graph $G$ which satisfies $\Buchi$.
\begin{itemize}
\item[$(i)$] We construct by transfinite recursion an increasing ordinal-indexed sequence of subsets of $V(G)$ by the formula
\[
V_\lambda = \big\{ v \in V(G) \mid v \re c v'\implies [c=\good \tor \exists \beta < \lambda, v' \in V_\beta] \big\}.
\]
Note that the definition is identical to that of the proof of Lemma~\ref{lem:non_uniform_construction_reachability}, thus we may skip a few steps below which were already proved.
\item[$(ii)$] We let $R=\bigcup_\lambda V_\lambda$ and prove that $R=V(G)$:
from $v_0 \notin R$, we may construct a path $v_0 \re \wait v_1 \re \wait \dots$ in $G$, which contradicts the fact that $G$ satisfies $\Buchi$.
\item[$(iii)$] It again holds that $V_\lambda = V_{\lambda+1}$ implies $V_{\lambda'} = V_{\lambda}$ for $\lambda' > \lambda$.
\item[$(iv)$] We let $\alpha$ be such that $|\alpha| > |V(G)|$ and we have $V_{\lambda}=V_{\lambda+1}$ for some $\lambda < \alpha$.
\item[$(v)$] Therefore $R = \bigcup_{\lambda < \alpha} V_\alpha = V(G)$ and we let $\phi:V(G) \to V(U_\alpha) = [0,\alpha)$ be given by $\phi(v) = \min \{\lambda \mid v \in V_\lambda\}$.
\item[$(vi)$] We verify that $\phi$ defines a graph morphism, which follows directly from the definitions.\qedhere
\end{itemize}
\end{proof}

\paragraph{Almost universal graphs.}
We now provide a general technique for constructing universal graphs, which we will then apply to the co-B\"uchi objective (and later, to other examples).
Fix a prefix-independent objective $W \subseteq C^\omega$.
Recall that for a vertex $v$ in a graph $G$, we use $G[v]$ to denote the restriction of $G$ to vertices reachable from $v$.
We say that a graph $U$ is almost $(\kappa,W)$-universal if 
\begin{itemize}
\item $U$ satisfies $W$; and
\item all graphs $G$ of cardinality $< \kappa$ satisfying $W$ have a vertex $v \in V(G)$ such that $G[v] \to U$.
\end{itemize}
When a graph is almost $(\kappa,W)$-universal for all cardinals $\kappa$, we say that it is uniformly almost $W$-universal.

Given a graph $U$ and an ordinal $\alpha$ we let $U \alpha$ be the graph defined by $V(U \alpha)= U \times \alpha$ and
\[
(u,\lambda) \re c (u',\lambda') \tin U \alpha \quad \iff \quad \lambda > \lambda' \tor (\lambda = \lambda' \tand u \re c u' \tin U). 
\]
Note that if $U$ is well-monotone, then so is $U\alpha$ (with respect to the lexicographical order on $U \times \alpha$.)
In the terminology of Section~\ref{sec:lexico}, $U\alpha$ is the lexicographical product of $U$ and the (well-monotone) edgeless graph over $\alpha$.

The following very helpful result reduces the search for a well-monotone universal graph to that of a well-monotone almost universal graph.

\begin{lemma}\label{lem:rongeur_de_croutes}
Let $U$ be an almost $(\kappa,W)$-universal graph and let $|\alpha| \geq \kappa$.
Then $U\alpha$ is $(\kappa,W)$-universal.
\end{lemma}

\begin{proof}
Consider an infinite path $(u_0,\lambda_0) \re {c_0} (u_1,\lambda_1) \re {c_1} \dots$ in $U\alpha$.
Since $\lambda_0 \geq \lambda_1 \geq \dots$, it must be that this sequence is eventually constant by well-foundedness.
Therefore, some suffix $u_i \re {c_i} u_{i+1} \re{c_{i+1}} \dots$ defines a path in some copy of $U$, which implies that $c_i c_{i+1} \dots \in W$.
We conclude by prefix independence that $U\alpha$ indeed satisfies $W$.

Let $G$ be a graph of cardinality $< \kappa$ which satisfies $W$.
We construct by transfinite recursion an ordinal sequence of vertices $v_0,v_1 \dots \in V(G)$, where for each $\beta < \lambda$, $v_\lambda$ is not reachable from $v_\beta$ in $G$, together with a morphism $\phi_\lambda : G_\lambda \to U$, where $G_\lambda$ is the restriction of $G$ to vertices reachable from $v_\lambda$ but not from $v_\beta$ for $\beta < \lambda$.

Assuming the $v_\beta$'s for $\beta < \lambda$ are already constructed (this assumption is vacuous for the base case $\lambda=0$), there are two cases.
If all vertices in $G$ are reachable from some $v_\beta$, then the process stops.
Otherwise, we let $G_{\geq \lambda}$ be the restriction of $G$ to vertices not reachable from any $v_\beta$ for $\beta < \lambda$.
It is a nonempty graph of cardinality $< \kappa$ satisfying $W$; we let $v_\lambda$ be such that $G_{\geq \lambda}[v_\lambda] = G_\lambda$ has a morphism $\phi_\lambda$ towards $U$.

Since all the $G_\lambda$'s are nonempty, the process must terminate in $\lambda_0$ steps for $|\lambda_0| \leq |G| < \kappa = \alpha$.
Now observe that any edge in $G$ is either from $G_\beta$ to itself, for some $\beta \leq \lambda_0 < \alpha$, or from $G_\beta$ to $G_{\beta'}$ for $\beta' < \beta \leq \lambda_0 < \alpha$.
This proves that the map $\phi:V(G) \to V(U\alpha)$ defined by $\phi(v)=(\phi_\lambda(v),\lambda)$ where $v \in V(G_\lambda)$ is a morphism from $G$ to $U\alpha$.
\end{proof}

\paragraph{Co-B\"uchi games.}
The co-B\"uchi condition is defined over $C=\{\safe, \bad\}$ by
\[
\Cobuchi = \{w \in C^\omega \mid |w|_{\bad} < \infty\} = C^* \safe^\omega.
\]
It is prefix-independent, thus we aim to construct well-monotone graphs which satisfy $\Cobuchi$ and embed graphs satisfying $\Cobuchi$.
Given an ordinal $\alpha$ consider the graph $U_\alpha$ given over $V(U_\alpha)=\alpha=[0,\alpha)$ by
\[
\lambda \re c \lambda' \tin U_\alpha \quad \iff \quad \begin{array}{lcl}
c = \bad & \tand & \lambda > \lambda' \qquad \tor \\
c = \safe & \tand & \lambda \geq \lambda'.
\end{array}
\]
It is depicted in Figure~\ref{fig:cobuchi_construction}.

\begin{figure}[ht]
\begin{center}
\includegraphics[width=0.6\linewidth]{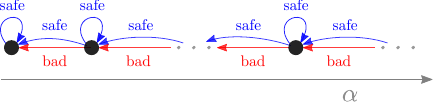}
\end{center}
\caption{The $\{\safe,\bad\}$-graph $U_\alpha$ defined with respect to the co-B\"uchi condition.}\label{fig:cobuchi_construction}
\end{figure}

\begin{lemma}\label{lem:non_uniform_construction_cobuchi}
For any ordinal $\alpha$, $U_\alpha$ is $(|\alpha|,\Cobuchi)$-universal.
\end{lemma}

We prove the result by applying the of almost universality technique outlined above.

\begin{proof}
Let $U$ be the well-monotone one-vertex graph with a $\safe$-loop; observe that $U_\alpha = U\alpha$.
Thus by Lemma~\ref{lem:rongeur_de_croutes}, it suffices to prove that $U$ is almost $(|\alpha|,\Cobuchi)$-universal.
We will in fact prove that $U$ is uniformly almost $\Cobuchi$-universal.
It is clear that $U$ satisfies $\Cobuchi$. 

Fix a graph $G$ satisfying $\Cobuchi$, and assume for contradiction that for all vertices $v \in V(G)$, there is a $\bad$-edge in $G[v]$.
Then one can construct by a quick induction a path with infinitely many $\bad$-edges in $G$; a contradiction.
Therefore, there is $v \in V(G)$ such that $G[v]$ has only $\safe$-edges, and thus the unique map $V(G[v]) \to V(U)$ defines a morphism $G \to U$, as required.
\end{proof}

\paragraph*{$K$-monotonicity.}
The $K$-monotone objective (for ``monotone'' in the sense of Kopczy\'nski) associated to a finite monotone $C$-graph $U$ is given by 
\[
    W = C^* W_0 \subseteq C^\omega,
\]
where $W_0$ is the set of colorations of infinite paths over $U$ (which, by monotone composition, coincides with the set of colorations from the maximal vertex in $U$).
Equivalently, $W$ is the set of colorations of the well-monotone graph $U \omega$.
Note that $W$ is $\omega$-regular and prefix-independent.

It is not hard to see that this definition corresponds to that of~\cite{Kopczynski06}.
The co-B\"uchi objective is an example, where $U$ is the one-vertex graph with a $\safe$-loop; one may generate many other examples by fixing $U$ to be any finite monotone graph (see Figure~\ref{fig:k-monotone}).

\begin{figure}[h]
\begin{center}
\includegraphics[width=0.72\linewidth]{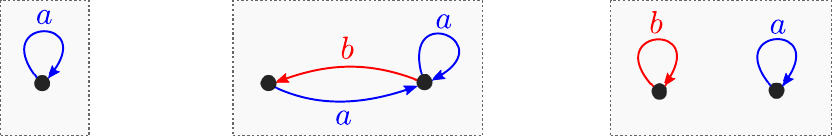}
\end{center}
\caption{Three finite $C$-monotone graphs for $C=\{a,b\}$; as always, edges following from monotone composition are not depicted.
The three graphs from left to right give rise to the three following $K$-monotone condition: $C^* a^\omega$ (a co-B\"uchi condition), $C^*(a + ba)^\omega$ (finitely many occurrences of the factor $bb$), and $C^*(a^\omega + b^\omega)$ (a union of two co-B\"uchi conditions).}\label{fig:k-monotone}
\end{figure}

Note that the union of two $K$-monotone objectives given by $U_1$ and $U_2$ is the $K$-monotone objective given by their directed sum $U_1 + U_2$ (see Section~\ref{sec:healing} for a formal definition of directed sums).

\begin{lemma}
Let $W \subseteq C^\omega$ be a $K$-monotone objective and let $U$ be the associated finite monotone $C$-graph.
For any ordinal $\alpha$, the well-monotone graph $U \alpha$ is $(|\alpha|,W)$-universal.
\end{lemma}

\begin{proof}
We show that $U$ is uniformly almost $W$-universal, which implies the result thanks to Lemma~\ref{lem:rongeur_de_croutes}.
First, it is clear that the coloration of any path in $U$ is a coloration in $U \omega$, thus $U \models W$.

We let $W_0 \subseteq W$ be the set of infinite colorations of $U$, and $W_0^\fin \subseteq C^*$ be the set of finite colorations of $U$.
Since $U$ is finite and has no sink, it holds that
\[
    \tag{$*$}
    w=w_0w_1 \dots \in W_0 \quad \iff \quad \forall k, w_0 \dots w_k \tin W_0^\fin.
\]
Now let $G$ be an arbitrary graph satisfying $W$.
We first claim that there must be a vertex $v \in V(G)$ such that $G[v] \models W_0$. 
Towards contradiction, assume otherwise: for all $v \in V(G)$, there is an infinite path $v \rp w$ such that $w \notin W_0$.
Thanks to $(*)$, this rewrites as a finite path $v \rp{w} v'$ with $w \notin W_0^\fin$.
Starting from any $v_0 \in V(G)$ we thus construct by a quick induction an infinite path
\[
    v_0 \rp {w_0} v_1 \rp {w_1} \dots \tin G,
\]
such that for all $i$, $w_i \notin W_0^\fin$.
Since $G \models W$, we then have $w_0 w_1 \dots \in W=C^* W_0$ thus for some $i$, $w_i w_{i+1} \dots$ is a coloration of $U$.
But this implies that $w_i \in W_0^\fin$, a contradiction.

Hence there is $v_0 \in V(G)$ such that $G[v_0] \models W_0$; there remains to define a morphism from $G[v_0]$ to $U$.
Note that for any $v \in G[v_0]$, it also holds that $G[v] \models W_0$.
Indeed, for any coloration~$w'$ from~$v'$, there is a finite word $w \in C^*$ (the coloration of a path from $v$ to $v'$) such that $ww'$ is a coloration from $v$ in $G$.
Thus $ww'$ is a coloration in $U$ (since $G[v] \models W_0$) and hence so is $w'$.

Now recall that if $u \geq u'$ in $U$, then by monotone composition, there are more colorations from $u$ than from $u'$ in $U$.
Let us then define a map $\phi: V(G[v_0]) \to V(U)$ by
\[
    \phi(v) = \min\{u \in U \mid \forall w \in C^*, v \rp w \tin G \implies u \rp w \tin U\}.
\]
In words, a vertex $v$ is mapped to the smallest position in $U$ which has all colorations from $v$.
Note that $\phi$ is well defined since for $v \in V(G[v_0])$ it holds that $G[v] \models W_0$; stated differently all colorations from vertices in $V(G[v_0])$ are colorations from $\max V(U)$ in $U$.

We now prove that $\phi$ defines a morphism $G[v_0] \to U$.
Let $e=v \re c v' \tin G[v_0]$.
Then for all colorations $w'$ from $v'$ in $G[v_0]$, $cw'$ is a coloration from $v$ in $G[v_0]$.
Therefore, $\phi(v)$ has all colorations $cw'$ in $U$ where $w'$ is a coloration from $v'$ in $U$.
Let $u'$ be the maximal $c$-successor of $\phi(v)$ in $U$; we prove that $\phi(v') \leq u'$ which implies that $\phi(v) \re c \phi(v')$ by monotone composition.

For this, we show that $u'$ has all colorations from $v'$, which gives the result by minimality of $\phi(v')$.
Let $w'$ be such a coloration.
We know that $cw'$ is a coloration from $\phi(v)$ in $U$: let $\phi(v) \re c u'' \rp{w'} \tin U$.
By maximality of $u'$ among $c$-successors of $\phi(v)$ we get $u' \geq u$, thus $u \rp{w'} \tin U$, as required.
This concludes the proof that $\phi$ is a morphism, thus $U$ is indeed uniformly almost $W$-universal, as required.
\end{proof}

Therefore $K$-monotone objectives are positional, as was established by Kopczy\'nski in~\cite{Kopczynski06}.

\subsection{Counter-based examples}\label{subsec:quantitative_examples}

We now discuss a few quantitative valuations.

\paragraph{Energy games.}
We start with the energy valuation, given over $C=\Z$ by
\[
\Energy(t_0t_1\dots)=\sup_{k} \sum_{i=0}^{k-1} t_i \in [0,\infty].
\]
We consider the graph $U$ over $V(U)=\omega$ given by
\[
u \re t u' \tin U \qquad \iff \qquad t \leq u - u' \in \Z.
\]
It is illustrated in Figure~\ref{fig:energy_construction}.

\begin{figure}[ht]
\begin{center}
\includegraphics[width=0.8\linewidth]{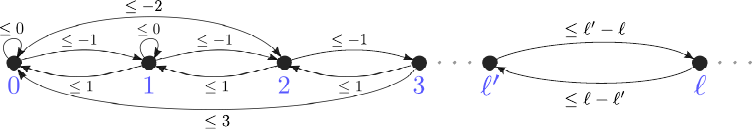}
\end{center}
\caption{The monotone $\Z$-graph $U$ corresponding to the $\Energy$ valuation. The names of the vertices are displayed in blue to improve readability.
Not all edges are depicted, we simply write $\re{\leq t}$ for the conjunction of $\re{t'}$ for all $t' \leq t$.}\label{fig:energy_construction}
\end{figure}

The usual order defines a well-order over $V(U)$ with respect to which monotonicity is easy to verify.
For each $u \in \omega$ the path $u \re {u} 0 \re 0 0 \re 0 \dots$ has value $u$, hence $\Energy_{U}(u) \geq u$.

Conversely consider an infinite path from $u_0 \in \omega$. It is of the form $u_0 \re {t_0} u_1 \re{t_1} u_2 \re {t_2} \dots$ with for all $i$, $t_i \leq u_{i} - u_{i+1}$.
Hence the partial sums are telescopic: we have for all $k$,
\[
\sum_{i=0}^{k-1} t_i \leq u_0 - u_k \leq u_0.
\]
Therefore it holds that for all $u \in \omega$ we have
\[
\Energy_U(u) = u.
\]
Energy games are similar to safety games in the sense that they have a uniformly universal well-monotone graph.

\begin{lemma}
\label{lem:uniform_construction_energy_games}
The well-monotone graph $U^{\top}$ is uniformly $\Energy$-universal.
\end{lemma}

\begin{proof}
    Consider a graph $G$.
    We see the values in $G$ as defining a map from $V(G)$ into $V(U^\top)$, formally
    \[
    \begin{array}{llcl}
    \Energy_G :& V(G) & \to & V(U^\top) \\
    & v & \mapsto & \Energy_G(v),
    \end{array}
    \]
    where we identify $\top$ to $\infty$.
    
    The fact that it is $\Energy$-preserving follows from the fact that $\Energy_U(u)=u$, proven above.
    We prove that it is a morphism: consider an edge $e = v \re t v'$ in $G$.
    If $\Energy_G(v')=\top$ then $\Energy(v)=\top$ thus $\Energy_G(v) \re t \Energy_G(v')$ in $U^{\top}$.
    
    We assume otherwise and let $\pi'$ be a path from $v'$ in $G$ with value $\Energy(\pi')=\Energy_G(v')$ which we denote by $x' \in \omega$ for simplicity.
    Then $e\pi'$ defines a path from $v$ in $G$, therefore 
    \[
    \Energy_G(v) \geq \Energy(e\pi') = \max(0,t+x') \geq t + x',
    \]
    which rewrites as
    \[
    t \leq \Energy_G(v) - \Energy_G(v'),
    \]
    the wanted result.
\end{proof}

This implies thanks to Theorem~\ref{thm:structure_implies_positionality} that arbitrary energy games are positional.
Somewhat surprisingly, it appears that this result had not been formally established before, Lemma 10 in~\cite{BFLMS08} is stated over finite graphs\footnote{Otherwise, the result would not hold in any case, since it includes the opponent.}, whereas Corollary 8 in~\cite{CFH14} applies only to graphs of finite degree.

However, the opponent in an energy game can require arbitrary memory even over countable graphs of degree 2 and with bounded weights (see Figure~\ref{fig:infinite_energy1}).

\begin{figure}[ht]
\begin{center}
\includegraphics[width=0.85\linewidth]{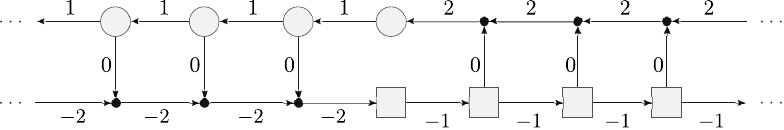}
\end{center}
\caption{An energy game in which Adam (squares) requires (infinite) memory to ensure value $\infty$ from any vertex.}\label{fig:infinite_energy1}
\end{figure}

\paragraph{Qualitative energy games.}

We make a quick detour via qualitative energy games, defined over $C=\Z$ by the condition
\[
    \QEnergy = \{t_0t_1\dots \in \Z^\omega \mid \sup_k \sum_{i=0}^{k-1} t_i < \infty\}.
\]
In a given graph $G$ and for a given vertex $v$, there is a discrepancy between $v$ having infinite quantitative energy (the supremum over paths from $v$ is bounded) and not satisfying the quantitative energy objective (all paths from $v$ have finite energy); although the former implies the latter.
Figure~\ref{fig:unbounded_energy} gives an example where $v$ satisfies $\QEnergy$ but has infinite energy value.

\begin{figure}[ht]
    \begin{center}
    \includegraphics[width=0.25\linewidth]{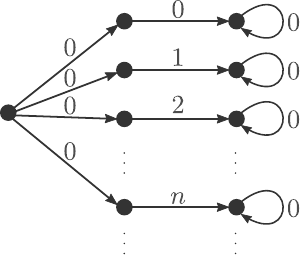}
    \end{center}
    \caption{A graph in which the vertex on the left has infinite quantitative energy, but satisfies the qualitative energy objective $\QEnergy$.}\label{fig:unbounded_energy}
\end{figure}

In particular, the graph $U^\top$ defined above is not $\QEnergy$-universal.
However, we prove that $U$ is uniformly almost $\QEnergy$-universal, which implies the following result.

\begin{lemma}
The well-monotone graph $U \alpha$ is $(|\alpha|, \QEnergy)$-universal.
\end{lemma}

\begin{proof}
We prove that $U$ is uniformly almost $\QEnergy$-universal, which implies the result by Lemma~\ref{lem:rongeur_de_croutes}.
We proved above that vertex $u \in V(U)=\omega$ has $\Energy$-value $u$ in $U$, therefore $U$ indeed satisfies $\QEnergy$.

Let $G$ be an arbitrary graph satisfying $\QEnergy$; we prove that there exists a vertex $v \in V(G)$ such that $\Energy_G(v) < \infty$.
Assume otherwise: for each $v \in V(G)$ and for each $N \in \omega$, there exists a finite path $\pi_v$ from $v$ whose weights sum up to $\geq N$.
Concatenating such paths we construct an infinite path with infinite energy in $G$, a contradiction.
Now if $v$ has finite energy in $G$, all vertices reachable from $v$ in $G$ also have finite energy, thus $G[v]$ embeds into $U$ by Lemma~\ref{lem:uniform_construction_energy_games}.
\end{proof}

Therefore qualitative energy games are also positional over arbitrary arenas.

\paragraph{Boundedness games.}

We now discuss boundedness games.
This class of games is defined over the set of colors
\[
    C = \{f : \omega \to \omega \mid f \text{ is monotone}\},
\]
by the objectives
\[
\BoundedCounter_N = \{f_0f_1 \dots \in C^\omega \mid \forall k, f_k(f_{k-1}(\dots(f_0(0))\dots)) \leq N\},
\]
where $N \in \omega$ is a fixed bound.
Intuitively, starting from $0$, a counter is updated along the path by applying function $f$ whenever an $f$-edge is seen.
Functions labelling edges are arbitrary monotone maps, for instance resetting, squaring, or raising to the next prime number are allowed.

Colcombet, Fijalkow and Horn~\cite{CFH14} have established that $\BoundedCounter_N$ is positionally determined over graphs of finite degree, we extend this result to arbitrary graphs.
Note that $\BoundedCounter_N$ is a prefix-increasing objective by monotonicity of the maps in $C$, therefore we are looking to construct a well-monotone graph which satisfies $\BoundedCounter_N$ and embeds small graphs satisfying $\BoundedCounter_N$.

We let $U_N$ be the graph over $V(U_N)=[0,N]$ given by
\[
u \re f u' \tin U_N \qquad \iff \qquad f(u) \leq u'.
\]
The graph $U$ is monotone with respect to the inverse order over $V(U_N)=[0,N]$, with minimal element $N$ and maximal element $0$.
It is well-monotone since all finite orders are well founded.
Note that fixing the bound $N$ is required for well foundedness; defining $U$ over $\omega$ as we did before fails when considering the dual ordering.

\begin{lemma}\label{lem:uniform_construction_boundedness}
For all $N \in \omega$, the well-monotone graph $U_N$ is uniformly $\BoundedCounter_N$-universal.
\end{lemma}

\begin{proof}
    We first show that $U_N$ satisfies $\BoundedCounter_N$: let $\pi : u_0 \re{f_0} u_1 \re{f_1} \dots$ be an infinite path in $U_N$.
    By definition it holds for all $i$ that $f_i(u_i) \leq u_{i+1}$ which implies by monotonicity that for all~$k$,
    \[
    f_k(f_{k-1}(\dots(f_0(0))\dots)) \leq u_{k+1} \leq N,
    \]
    the wanted result.
    
    We define a valuation\footnote{Our proof actually shows that $U_N^\top$ is universal with respect to this valuation, which is a bit more precise than the statement of the theorem.}
    \[
    \begin{array}{lrcl}
    \val_N :& C^\omega &\to& [0,N] \cup \{\bot\} \\
    & f_0 f_1 \dots & \mapsto & \max \{i \in [0,N] \mid \forall k, f_k(f_{k-1}(\dots(f_0(i))\dots)) \leq N\}.
    \end{array}
    \]

    The (complete) linear order over $[0,N] \cup \{\bot\}$ is again the reverse order, in particular $\bot$ is the maximal element, and should be thought of as ``right after zero''.
    For clarity, we still use $\geq$, $\min$ and $\max$ for the usual ordering over integers; it is understood in the definition above that $\max \varnothing = \bot$.
    
    Consider a $C$-graph $G$ which satisfies $\BoundedCounter_N$, we prove that $\val_G : V(G) \to [0,N]$ which by definition assigns $\min_{v \rp w} \val(w)$ to $v \in V(G)$, defines a morphism $G \to U_N$.
    Let $e_0=v_0 \re{f_0} v_1$ in $G$ and let $\pi_1 = v_1 \re{f_1} v_2 \re{f_2} \dots$ be an infinite path from $v_1$ in $G$ with minimal valuation $i_1 = \val(\pi_1)=\val_G(v_1)$.
    
    Then $\pi_0=e_0 \pi_1$ is a path from $v_0$ in $G$ thus $\val_G(v) \leq \val(\pi_0)$ which we denote by $i_0$.
    Note that both $i_0$ and $i_1$ are $\geq 0$ since $G$ satisfies $\BoundedCounter_N$.
    We have by definition
    \[
    i_0 = \max\{i \in [0,N] \mid \forall n, f_n(f_{n-1}(\dots(f_0(i))\dots)) \leq N\},
    \]
    hence for all $n$ it holds that $f_n(f_{n-1}(\dots(f_0(i_0))\dots)) \leq N$.
    Since
    \[
    i_1 = \max\{i \in [0,N] \mid \forall n, f_n(f_{n-1}(\dots(f_1(i))\dots)) \leq N\},
    \]
    we have in particular that $f_0(i_0) \leq i_1 = \val_G(v_1)$.
    By monotonicity of $f_0$, we then obtain $f_0(\val_G(v_0)) \leq \val_G(v_1)$, thus 
    \[
    \val_G(v_0) \re{f_0} \val_G(v_1)
    \]
    belongs to $U_N$, which concludes the proof.
    \end{proof}

This establishes that for all $N \in \omega$, $\BoundedCounter_N$ is positional (over arbitrary arenas).

\subsection{Finitary parity games}\label{subsec:finitary_parity}

We now study finitary parity games, introduced by Chatterjee, Henzinger and Horn~\cite{CHH08}.
It was shown by Chatterjee and Fijalkow~\cite{CF13} that finitary Büchi games are positional over arbitrary game graphs, and that finitary parity games are finite-memory determined.
We strengthen this result by establishing that finitary parity games are in fact positional over arbitrary game graphs.

\paragraph*{Definitions.}
Throughout this section, we let $d \in \omega$ be a fixed even number, and let $C=[0,d] \subseteq \omega$.
In this context, colors are usually called priorities.
Let $w=p_0p_1 \dots \in [0,d]^\omega$.
An occurrence $i$ of an odd priority $p_i$ in $w$ is called a request.
We say that a request $i$ is granted at time 
\[
    \min\{i' > i \mid p_i' \text{ is even and } \geq p_i\} \in [i+1,\infty].
\]
For a request $i$ granted at time $i'$, we let $t_i = i' - i \in [1,\infty]$ denote the (potentially infinite) delay within which the request is granted.
We then let $\req(w)$ be the finite or infinite word over $[1,\infty]$ obtained by concatenating all $t_i$'s (in order), where $i$ is a request in $w$.

Let us give a few examples:
\[
    \begin{array}{lcl}
        \req(33112^\omega) & = & \infty \infty 2 1 \\
        \req(12102100210002100002...) & = & 1 2 3 4 \dots \\
        \req(3(21)^\omega) & = & \infty 1^\omega.
    \end{array}
\]
Given a fixed bound $N \geq 1$, we let $\Grant_N$ be the objective comprised of words where all requests are granted within time $\leq N$, formally
\[
    \Grant_N = \{w \in [0,d]^\omega \mid \sup \req(w) \leq N\}.
\]
We also let
\[
    \Grant = \{w \in [0,d]^\omega \mid \sup \req(w) < \infty\} = \bigcup_{N \in \omega} \Grant_N.
\]
Objectives $\Grant_N$ as well as $\Grant$ are prefix-increasing, but not prefix-decreasing (since one may append a request which is never granted).
Recall that in this case, we look for (well-monotone) universal graphs in the sense of Colcombet and Fijalkow (satisfying the objective, and embedding small graphs satisfying the objective).
The finitary parity objective is the prefix-independent variant of $\Grant$, defined by 
\[
    \FinParity = [0,d]^* \Grant = \{w \in [0,d]^\omega \mid \req(w) \text{ is finite } \tor \limsup \req(w) < \infty\}.
\]
In words, Eve should ensure that eventually, all requests are granted.

\begin{figure}[h]
\begin{center}
\includegraphics[width = 0.6\linewidth]{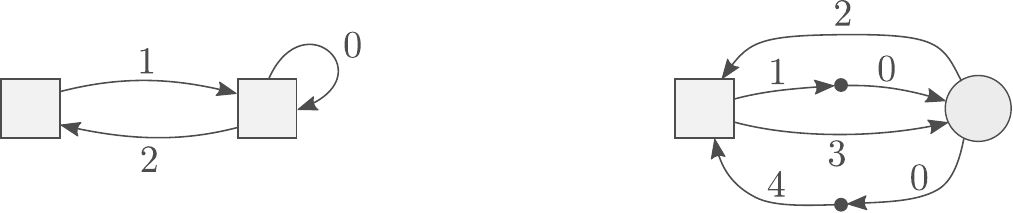}
\end{center}
\caption{On the left, a game where Eve loses the finitary parity game even though she wins the parity game (see Section~\ref{sec:lexico} for a definition of parity games). On the right, a game for which Eve wins the objective $\Grant_2$, but cannot do so with a positional strategy.}\label{fig:finpar1}
\end{figure}

\paragraph*{Constructions.}
We define a finite graph $U_N$ over 
\[
    V(U_N) = \{(m,p) \in [1,N] \times [0,d] \mid p \text{ is odd}\} \cup \{(0,0)\}.
\]
It is illustrated in Figure~\ref{fig:construction_fin_parity}.

Let us first define a linear order over $V(U_N)$ as follows:
\[
    (m,p) \geq (m',p') \quad \iff \quad p < p' \tor [p' = p \tand m \leq m'].
\]
Then edges in $U_N$ are defined by
\[
    (m,p) \re q (m',p') \tin U_N \quad \iff \quad \begin{array}{llr}
        & q \text{ is even and } \geq p \qquad\qquad\qquad\qquad\qquad\qquad & (a) \\
        \tor & q \leq p \tand (m,p) > (m',p') & (b)\\
        \tor & p' \geq q > p. & (c)
    \end{array}
\]

\begin{figure}[h]
    \begin{center}
    \includegraphics[width=\linewidth]{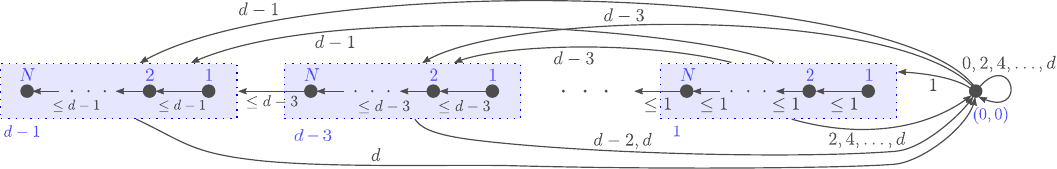}
    \end{center}
    \caption{An illustration of the graph $U_N$; not all edges are depicted for readability.}\label{fig:construction_fin_parity}
\end{figure}

For convenience, we call these edges of type $(a),(b)$ or $(c)$ (formally, an edge can have several types at once).
Note that there are $d$-edges (of type $(a)$) between each pair of vertices, in particular $U_N$ is indeed a graph.

\begin{lemma}
For any $N \geq 1$, the graph $U_N$ is monotone.
\end{lemma}

It is not hard to see on Figure~\ref{fig:construction_fin_parity} that the lemma holds; we spell out a proof for completeness.

\begin{proof}
First, it is a direct check that $(m,p) \re q (m',p') > (m'',p'')$ implies $(m,p) \re q (m'',p'')$ in $U_N$ (it even holds that the type of the edge is preserved).
The second case is more tedious.
Assume $(m,p) > (m',p') \re q (m'',p'')$; we aim to prove that $(m,p) \re q (m'',p'') \tin U_N$.
\begin{itemize}
    \item If $q$ is even and $\geq p'$. Since $p' \geq p$, we have an edge $(m,p) \re q (m'',p'')$ of type $(a)$.
    \item If $q \leq p'$ and $(m',p') > (m'',p'')$. If $q \leq p$ then since $(m,p)>(m',p')>(m'',p'')$, we have an edge $(m,p) \re q (m'',p'')$ of type $(b)$. Otherwise, we have $p'' \geq p' \geq q > p$, so $(m,p) \re q (m'',p'')$ is an edge of type $(c)$.
    \item If $p'' \geq q > p'$ then we also have an edge of type $(c)$ since $p' \geq p$.\qedhere
\end{itemize}
\end{proof}

We now prove the crucial result.
In some sense, it states that $U_N$ is a good approximation of a $\Grant_N$-universal graph.

\begin{lemma}\label{lem:approx_univ}
Let $N \geq 1$. It holds that
\begin{itemize}
\item $U_N \models \Grant_{dN/2}$; and
\item for all graphs $G$ satisfying $\Grant_N$ there is a morphism $G \to U_N$.
\end{itemize}
\end{lemma}

Before proving the lemma, let us remark that one cannot hope to find well-monotone $\Grant_{N}$-universal graphs.
The reason is that $\Grant_N$ is not positional, an example witnessing non-positionality is drawn in Figure~\ref{fig:finpar1}.
Therefore, having a more relaxed approximation statement as above is required.

\begin{proof}
We prove the two items separately, starting from the first.
We must show that for all paths in $U_N$, requests are granted within at most $dN$ steps.
Consider a finite path $(m_0,p_0) \re {q_0} \dots \re {q_{k-1}} (m_k,p_k)$ in $U_N$ such that $q_0$ is odd and none of $q_1,...,q_{k-1}$ are even and $\geq q_0$ (that is, a request is opened and not closed after $k-1$ time steps).
We prove that $k \leq dN/2$.
Remark that for any edge $(m,p) \re q (m',p')$ which is not of type $(a)$ we have $p' \geq q$ and $(m,p) > (m',p')$.

Let us prove by a quick induction that for all $i \in \{1, \dots, k-1\}$, $p_i$ is $\geq q_0$ and $(m_i,p_i) > (m_{i+1},p_{i+1})$.
This is clear for $p_1$ since $(m_0,p_0) \re {q_0} (m_1,p_1)$ is not of type $(a)$.
Now assuming $p_i \geq q_0$, since moreover $q_i$ is not [even and $\geq q_0$], the edge $(m_i,p_i) \re {q_i} (m_{i+1},q_{i+1})$ is not of type~$(a)$, therefore we conclude by the above remark.

Therefore $(m_0,p_0) > (m_1,p_1) > \dots > (m_k,p_k)$ which implies that $k + 1 \leq |U_N| = dN/2+1$, the wanted result.

\vskip1em
We now prove the second item.
Let $G$ be a graph satisfying $\Grant_N$; we aim to construct a morphism $G \to U_N$.
Intuitively, this requires understanding, given $v \in V(G)$, what requests can be currently opened on a path when it visits $v$, and therefore we should look at the past of $v$.

For an odd priority $p$, we say that $p$ can be open at $v \in V(G)$ if there exists a path $v_0 \re {p_0} v_1 \re {p_1} \dots \re {p_{m-1}} v_m = v$ such that $p_0 = p$ and none of $p_1, \dots, p_{k-1}$ are even and $\geq p_0$.
Such a path is called a realizing path for $p$ and $v$.
Note that the length of a realizing path cannot exceed $N$ since $G$ satisfies $\Grant_N$.

We now define $\phi:V(G) \to V(U_N)$ by setting $\phi(v)$ to be $(0,0)$ if no odd priority can be open at $v$, and otherwise $\phi(v) = (m,p)$, where $p$ is the maximal odd priority that can be open at $v$, and $m$ is the maximal size of a realizing path for $p$ and $v$.
Stated differently, $\phi(v)$ is the minimal (with respect to the order over $V(U_N)$) $(m,p)$ such that $v$ has a realizable path of length $m$ for $p$, and $(0,0)$ if there are none.
We claim that $\phi$ defines a morphism $G \to U_N$.

Let $e = v \re q v'$ be an edge in $G$ and denote $\phi(v) = (m,p)$ and $\phi(v')=(m',p')$.
We let $\pi = v_0 \re {p_0=p} \re \dots \re {p_{m-1}} v_m=v$ be a minimal realizable path for $p$ and $v$ if it exists (that is, if $(m,p) \neq (0,0)$), and the empty path otherwise.
We distinguish three cases.
\begin{itemize}
\item If $q$ is even and $\geq p$, then $(m,p) \re q (m',p')$ is an edge of type $(a)$ in $U_N$, regardless of $(m',p')$.
\item If $q \leq p$, then $\pi e$ is a realizable path for $p$ and $v'$ of length $m+1$, therefore $m+1 \leq N$ and $(m,p) > (m+1,p) \geq \phi(v')$, thus $(m,p) \re q (m',p')$ is an edge of type $(b)$.
\item If $q>p$ and $q$ and we are not in the first case, then $q$ must be odd therefore $e$ is a realizable path for $q$ and $v'$ of length $1$.
Hence $p' \geq q > p$, and $(m,p) > (m+1,p)$ is an edge of type~$(c)$. \qedhere
\end{itemize}
\end{proof}

We now let $U$ be the directed sum of the $U_N$'s, formally it is defined over
\[
    \begin{aligned}
    V(U) &= \{(m,p,N) \in \omega^3 \mid [p \text{ is odd and } 1 \leq m \leq N] \tor [m=p=0 \tand N \geq 1]\} \\
    &= \bigsqcup_{N \geq 1} V(U_N) \times \{N\},
    \end{aligned}
\]
which is well-ordered by the lexicographic product over $V(U_N) \times [1,\infty)$, and given by edges
\[
    (m,p,N) \re q (m',p',N') \quad \iff \quad N>N' \tor [N=N' \tand (m,p) \re q (m',p')\tin U_N].
\]
Monotonicity of $U$ follows directly from monotonicity of the $U_N$'s.
We have the following result.

\begin{theorem}
    For every ordinal $\alpha$, the well-monotone graph $U\alpha$ is $(|\alpha|, \FinParity)$-universal.
\end{theorem}

This implies that $\FinParity$ is positional thanks to Theorem~\ref{thm:structure_implies_positionality}.

\begin{proof}
We prove that $U$ is uniformly almost $\FinParity$-universal, and the result follows from Lemma~\ref{lem:rongeur_de_croutes}.
First, let us show that $U \models \FinParity$: an infinite path in $U$ has a suffix in some $U_N$, which satisfies $\Grant_{dN/2}$ by Lemma~\ref{lem:approx_univ}; therefore it satisfies $\FinParity$ by prefix-independence.

Pick an arbitrary graph $G$ satisfying $\FinParity$, and assume towards contradiction that for all $v \in V(G)$ and for all $N \geq 1$, $G[v]$ does not satisfy $\Grant_N$: there is a vertex $v'$ reachable from~$v$ and a path
\[
    \label{eq:3}\tag{$*$}
    v'=v_0 \re {p_0} v_1 \re {p_1} \dots \re{p_N} v_{N+1}
\]
in $G$ such that $p_0$ is odd and $p_1,\dots,p_N$ are not [even and $\geq p_0$].
Then by concatenating such paths, we may easily construct a path in $G$ which does not satisfy $\FinParity$, we now give details for completeness.

Start from any $v_0 \in V(G)$, and let $v'_0$ be reachable from $v_0$ in $G$, say $v_0 \rp{w_0} v'_0$, and $v'_0 \rp{w'_0} v_1$ be a path as in $(*)$ with $N=1$.
Then iterate from $v_1$ with $N=2$: we get $v_1 \rp{w_1} v'_1 \rp{w'_1} v_2$ where the first letter of $w'_1$ defines an occurrence which is not closed after $2$ steps.
We continue this process for $N=3,4,\dots$ and get an infinite path coloration $w=w_0 w'_0 w_1 w'_1 \dots $ in $G$ such that $\limsup \req(w) = \infty$, a contradiction.

Therefore there is $v \in V(G)$ and $N \geq 1$ such that $G[v] \models \Grant_N$, so by Lemma~\ref{lem:approx_univ}, $G[v] \to U_{N} \to U$.
We conclude that $U$ is uniformly almost $\FinParity$-universal as required.
\end{proof}

\section{Closure Under Lexicographical Products}\label{sec:lexico}

In this section, we show how our characterization can be exploited to derive a new closure property, which applies to prefix-independent objectives.

\subsection{Definitions and statement of the result}

Let us start by defining lexicographical products of prefix-independent objectives (this definition does not make sense for prefix-dependent objectives).

\paragraph{Product of objectives.}
We consider two prefix-independent objectives $W_1 \subseteq C_1^\omega$ and $W_2 \subseteq C_2^\omega$, where $C_1$ and $C_2$ are disjoint sets of colors.
We let $C=C_1 \sqcup C_2$ and for $w \in C^\omega$ we let $w_1 \in C_1^{\leq \omega}$ and $w_2 \in C_2^{\leq \omega}$ be the finite or infinite words obtained by restricting $w$ to colors in $C_1$ or in $C_2$, respectively.
Note that if $w_2$ is finite then $w_1$ is infinite.

We define the \defin{lexicographical product} of $W_1$ and $W_2$ by
\[
W_1 \lex W_2= \left\{w \in C^\omega \left| \begin{array}{l}
w_2 \text{ is infinite and } w_2 \in W_2 \tor \\
w_2 \text{ is finite and } w_1 \in W_1
\end{array}
\right.
\right\}.
\]
Let us stress the fact that this operation is not commutative; intuitively, more importance is given here to $W_2$.
Note that this operation is however associative, and given $W_1,\dots,W_h$ respectively over pairwise disjoint sets of colors $C_1,\dots, C_h$, we have
\[
    W_1 \lex \dots \lex W_h = \{w \in C^\omega \mid w_p \in W_p \text{ where } p \text{ is maximal such that } w_p \text{ is infinite}\},
\]
where $C = C_1 \sqcup \dots \sqcup C_h$.

\paragraph{Product of monotone graphs.}
We now consider two well-monotone graphs $G_1$ and $G_2$, respectively over colors $C_1$ and $C_2$.
We define their lexicographical product $G = G_1 \lex G_2$ to be given by $V(G) = V(G_1) \times V(G_2)$ and for all $c_2 \in C_2$ and $c_1 \in C_1$,
\[
    \begin{array}{rcl}
    (v_1,v_2) \re{c_2} (v'_1,v'_2) \tin G & \quad \iff \quad & v_2 \re{c_2} v_2' \tin G_2 \\
    (v_1,v_2) \re{c_1} (v'_1,v'_2) \tin G & \quad \iff \quad & v_2 > v_2'  \tor [v_2 = v_2' \tand v_1 \re{c_1} v_1' \tin G_1]
    \end{array}
\]
Naturally, $V(G)$ is equipped with the lexicographical well-ordering given by
\[
(v_1,v_2) \geq (v'_1,v'_2) \quad \iff \quad \begin{array}{l}
v_2 > v'_2  \tor [v_2 = v'_2 \tand v_1 \geq v'_1].
\end{array}
\]
As expected, we have the following result, which is a direct check.

\begin{lemma}
The graph $G=G_1 \lex G_2$ is well-monotone.
\end{lemma}

\paragraph*{Main result.}
We may now state our main result in this section; it is proved at the end of the section.

\begin{theorem}\label{thm:universality_lexicographical_product}
Let $W_1 \subseteq C_1^\omega$, $W_2 \subseteq C_2^\omega$ be prefix-independent conditions with $C_1 \cap C_2 = \varnothing$, let $\kappa$ be a cardinal number, and assume that the graphs $U_1$ and $U_2$ are $(\kappa,W_1)$ and $(\kappa,W_2)$-universal, respectively.
Then $U_1 \lex U_2$ is $(\kappa,W_1 \lex W_2)$-universal.
\end{theorem}

In the statement above, it is actually required that $U_1$ and $U_2$ also embed all pregraphs of size $< \kappa$ satisfying $W_1$ and $W_2$.
By Lemma~\ref{lem:pregraphs}, for nonempty graphs and thanks to prefix-independence this does not make any difference (nonempty universal graphs also embed pregraphs).
However, there is a subtlety when it comes to the trivially losing condition (which has an empty universal graph), details are discussed below.

We obtain the wanted closure property by combining Theorem~\ref{thm:universality_lexicographical_product} with our characterization result.

\begin{corollary}\label{cor:closure_lexico}
The class of positional prefix-independent objectives admitting a neutral color is closed under lexicographical product.
\end{corollary}

\begin{proof}[Proof of Corollary~\ref{cor:closure_lexico}]
Let $W_1$ and $W_2$ be such objectives and let $\kappa$ be a cardinal.
By Corollary~\ref{cor:positionality_implies_structure}, there are well-monotone graphs $U_1$ and $U_2$ which are respectively $(\kappa,W_1)$ and $(\kappa,W_2)$-universal.
By Theorem~\ref{thm:universality_lexicographical_product}, their lexicographical product $U_1 \lex U_2$ is $(\kappa,W_1 \lex W_2)$-universal.
By Theorem~\ref{thm:structure_implies_positionality}, this yields positionality of $W_1 \lex W_2$.
Note that the neutral color $\eps_1  \in C_1$ with respect to $W_1$ is also neutral with respect to $W$.
\end{proof}

To the best of our knowledge, this is the only known closure property for positional objectives.
We do not know whether it is possible to derive Corollary~\ref{cor:closure_lexico} without introducing well-monotone graphs.

\subsection{Examples of lexicographic products}

Before going on to the proof, we discuss a few examples.
First, observe that our notation $U\alpha$ coincides with the lexicographical product $U \otimes \alpha$, where $\alpha$ is the edgeless ordered pregraph over $\alpha$.

\paragraph{The two trivial conditions.}
Let us discuss the two conditions over one-letter alphabets.
We let $\TW(c)$ and $\TL(c)$ respectively denote the trivially winning and losing conditions over $C=\{c\}$, given by
\[
    \TW(c) = \{c^\omega\} \qquad \tand \qquad \TL(c) = \varnothing.
\]
Note that both are prefix-independent.
Consider the one-vertex graph $U_{\TW(c)}$ with a $c$-loop.
It is well-monotone, satisfies $\TW(c)$, and embeds all $\{c\}$-graphs (which all satisfy $\TW(c)$), so it is uniformly $\TW(c)$-universal.

\begin{figure}[h]
    \begin{center}
    \includegraphics[width=0.7\linewidth]{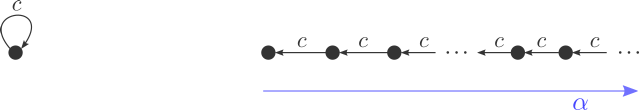}
    \end{center}
    \caption{On the left, the graph $U_{\TW(c)}$; on the right, the pregraph $U_{\TL(c)}$.}\label{fig:trivial_objectives}
\end{figure}

For the trivially losing condition, the situation is slightly more complex.
There are no graphs satisfying $\TL(c)$, therefore the empty graph is uniformly $\TL(c)$-universal.
However in the context of lexicographical products we require working with pregraphs.
Recall that Lemma~\ref{lem:pregraphs} states that non-empty universal graphs also embed pregraphs satisfying the condition; however this is not the case of the empty graph.
Indeed, there are non-empty pregraphs satisfying $\TL(c)$, namely, a pregraph satisfies $\TL(c)$ if and only if it has no infinite path.

Given an ordinal $\alpha$, we let $U_{\TL(c),\alpha}$ be the pregraph over $V(U_{\TL(c),\alpha})=\alpha$ with edges
\[
    \lambda \re c \lambda' \quad \iff \quad \lambda > \lambda'.
\]
In particular, note that $0 \in V(U_{\TL(c),\alpha})$ is a sink.
Since $U_{\TL(c),\alpha}$ has no infinite path, it satisfies the trivially losing condition $\TL(c)$.

Moreover, it embeds all pregraphs of size $< |\alpha|$ satisfying $\TL(c)$.
Indeed, a pregraph satisfying $\TL(c)$ has a sink, therefore the edgeless 1-vertex graph $1$ is uniformly almost $\TL(c)$-universal, and the result follows by Lemma~\ref{lem:rongeur_de_croutes} since $U_{\TL(c),\alpha} = 1 \alpha$.

\paragraph{Back to B\"uchi and co-B\"uchi games.}
Observe that the lexicographical products $\TL(\wait) \lex \TW(\good)$ and $\TW(\safe) \lex \TL(\bad)$ respectively coincide with the B\"uchi and co-B\"uchi objectives defined in Section~\ref{sec:examples}.
Therefore, by Theorem~\ref{thm:universality_lexicographical_product}, the graphs $U_{\TL(\wait),\alpha} \lex U_{\TW(\good)}$ and $U_{\TW(\safe)} \lex U_{\TL(\bad),\alpha}$ are respectively $(|\alpha|, \Buchi)$ and $(|\alpha|,\Cobuchi)$-universal.
See Figure~\ref{fig:buchi_cobuchi} for an illustration of this discussion.

\begin{figure}[h]
    \begin{center}
    \includegraphics[width=\linewidth]{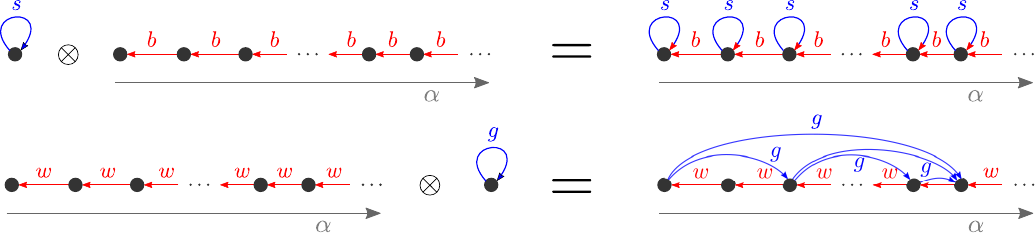}
    \end{center}
    \caption{Re-obtaining universal construction for co-B\"uchi and B\"uchi objectives by lexicographic product of trivial conditions.}\label{fig:buchi_cobuchi}
\end{figure}

Actually, these lexicographic products of well-monotone graphs correspond to the constructions given in Section~\ref{sec:examples} for B\"uchi and co-B\"uchi games.
Hence Theorem~\ref{thm:universality_lexicographical_product} gives alternative proofs for Lemmas~\ref{lem:non_uniform_construction_buchi} and~\ref{lem:non_uniform_construction_cobuchi}.

\paragraph{Parity games.}
Let $d \in \omega$ be a fixed even number.
The parity condition is defined over $C=[0,d] \subseteq \Z$ by
\[
    \begin{aligned}
    \Parity &= \{w \in [0,d]^\omega \mid \limsup w \text{ is even}\} \\
            &= \TW(0) \lex \TL(1) \lex \TW(2) \lex \dots \lex \TL(d-1) \lex \TW(d).
    \end{aligned}
\]
Therefore by Theorem~\ref{thm:universality_lexicographical_product}, for all ordinals $\alpha$ the graph
\[
    U_{\alpha} = U_{\TW(0)} \lex U_{\TL(1),\alpha} \lex U_{\TW(2)} \lex \dots \lex U_{\TL(d-1),\alpha} \lex U_{\TW(d)}
\]
is $(|\alpha|,\Parity)$-universal.
This gives a proof of positionality of parity games.

Unravelling the definitions, $U_{\alpha}$ is the graph over $\alpha^{d/2}$ given by
\[
    \begin{array}{rcl}
        (\lambda_1, \dots, \lambda_{d/2}) \re {2k-1} (\lambda'_1, \dots, \lambda'_{d/2}) \tin U_\alpha & \quad \iff \quad & (\lambda_k, \dots, \lambda_{d/2}) > (\lambda'_k, \dots, \lambda'_{d/2}) \\
        (\lambda_1, \dots, \lambda_{d/2}) \re {2k} (\lambda'_1, \dots, \lambda'_{d/2}) \tin U_\alpha & \quad \iff \quad & (\lambda_{k+1}, \dots, \lambda_{d/2}) \geq (\lambda'_{k+1}, \dots, \lambda'_{d/2}).
    \end{array}
\]
Therefore, our proof coincides with the proof of positionality of parity games from Emerson and Jutla~\cite{EJ91}, and morphisms into $U_\alpha$ correspond to signature assignments in the sense of Walukiewicz~\cite{Walukiewicz96}.

\subsection{Proof of Theorem~\ref{thm:universality_lexicographical_product}}\label{subsec:proof_lexico}

We will make use of the following technical lemma.

\begin{lemma}\label{lem:pointwise_min}
Let $G$ be a $C$-pregraph, $U$ be a $C$-monotone graph, and $\phi,\phi'$ be two morphisms from~$G$ to~$U$.
The pointwise minimum $\psi$ of $\phi$ and $\phi'$ defines a morphism $G \to U$.
\end{lemma}

In particular, there is a well-defined (pointwise) minimal morphism from $G$ to $U$, for any well-monotone $U$.

\begin{proof}
Let $v \re c v' \tin G$, we prove that $\psi(v) \re c \psi(v') \tin U$.
Without loss of generality, we assume that $\psi(v)=\phi(v)$.
Then we have
\[
\psi(v) = \phi(v) \re {c} \phi(v') \geq \psi(v')  \tin U, 
\]
and the result follows by monotone composition.
\end{proof}

Let us now fix a cardinal $\kappa$ and two well-monotone graphs $U_1$ and $U_2$ which are $\kappa$-universal with respect to $W_1 \subseteq C_1^\omega$ and $W_2 \subseteq C_2^\omega$, respectively.
We denote $U = U_1 \lex U_2$ and $C = C_1 \sqcup C_2$.
Recall that $W_1$ and $W_2$ are assumed to be prefix-independent, and therefore so is their lexicographical product $W = W_1 \lex W_2$.
There are two things to show: that $U$ satisfies $W$ and that $U$ embeds all graphs of cardinality $< \kappa$ which satisfy $W$.
We start with the first property.

\begin{lemma}
It holds that $U$ satisfies $W$.
\end{lemma}

\begin{proof}
Consider an infinite path $u^0 \re{c^0} u^1 \re{c^1} \dots$ in $U$, and let $w=c^0c^1 \dots$ be its coloration.
For all $i$ we let $u^i=(u_1^i,u_2^i)$ with $u_1^i \in U_1$ and $u_2^i \in U_2$.
 
Assume first that there are finitely many $c^i$'s which belong to $C_2$, and let $i_0$ be such that for all $i \geq i_0$, $c^i \in C_1$.
Then for all $i \geq i_0$, we have by definition that $u_2^i \geq u_2^{i+1}$.
Since $U_2$ is well-ordered, there is $i_1 \geq i_0$ such that for all $i \geq i_1$ we have $u_2^{i} = u_2^{i_1}$.
Therefore by definition of $U$, we have for all $i \geq i_0$ that $u_1^{i} \re{c^i} u_1^{i+1}$, and thus since $U_1$ satisfies $W_1$ and $W_1$ is prefix-independent, $w_1 \in W_1$.
Hence in this case, $w \in W$.

We now assume that there are infinitely many indices $i$ such that  $c^{i} \in C_2$, and let $i_0<i_1<\dots$ denote exactly those indices.
Then we have for all $j$ that all $c^i$'s with $i \in [i_j+1,i_{j+1}-1]$ belong to $C_1$ and thus by definition of $U$ it holds that $u_2^{i_j+1} \geq u_2^{i_{j+1}}$.
Hence we have in $U_2$
\[
    u_2^{i_0} \re{c^{i_0}} u_2^{i_0+1} \geq u_2^{i_1} \re{c^{i_1}} u_2^{i_1+1} \geq \dots,
\]
and thus by monotone composition in $U_2$,
\[
    u_2^{i_0} \re{c^{i_0}} u_2^{i_1} \re{c^{i_1}} \dots
\]
is a path in $U_2$.
Since $U_2$ satisfies $W_2$, we conclude that $w_2 \in W_2$ and thus $w \in W$.
\end{proof}

We now show the second property, namely that $U$ embeds all $C$-graphs of cardinality $< \kappa$ which satisfy $W$.
Let $G$ be such a graph.
We define a $C_2$-graph $G_2$ by $V(G_2) = V(G)$ and
\[
v \re {c_2} v' \tin G_2 \iff \exists v_1,v_1' \in V(G), v \rp{C_1^*} v_1 \re {c_2} v_1' \rp{C_1^*} v' \tin G,
\]
where the notation $\rp{C_1^*}$ refers to a finite path with colors only in $C_1$.
Note that $G_2$ may not be a graph: vertices which do not have a path visiting an edge with a color in $C_2$ are sinks in $G_2$.
This is why we require universality with respect to pregraphs.

\begin{lemma}
The pregraph $G_2$ satisfies $W_2$.
\end{lemma}

\begin{proof}
Consider an infinite path in $G_2$; it is of the form
\[
\pi_2 : v_0 \re {c_1} v_3 \re {c_4} v_6 \re{c_7} \dots
\]
where $c_1,c_4,c_7,\dots \in C_{2}$, and there exist vertices $v_1,v_2,v_4,v_5,v_7,v_8,\dots \in V(G)$ and finite words $w_0,w_2,w_3,w_5,w_6,w_8,\dots \in C_{1}^{*}$ such that
\[
\pi: v_0 \rp{w_0} v_1 \re {c_1} v_2 \rp{w_2} v_3 \re{w_3} v_{4} \re{c_4} v_5 \rp{w_5} v_6 \rp{w_6} v_7 \re{c_7} v_8 \rp{w_8} \dots
\]
defines a path in $G$.
Therefore $\col(\pi) = w \in W$, and since $w_2 = c_1 c_4 c_7 \dots = \col(\pi_2)$ is infinite this implies that $w_2 \in W_2$, the wanted result.
\end{proof}

We now consider the pointwise minimal morphism $\phi_2:G_2 \to U_2$.
Now comes the crucial technical claim.

\begin{lemma}\label{lem:no_small_edge_go_up}
If $v,v' \in V(G)$ are such that $\phi_2(v) < \phi_2(v')$ then there is no edge $v \re {c_1} v'$ in $G$ with $c_1 \in C_1$.
\end{lemma}

\begin{proof}
Assume for contradiction that there is such an edge.
Then in $G_2$, for all $c_2 \in C_{2}$, any $c_2$-successor of $v'$ is also a $c_2$-successor of $v$.
Consider the map $\phi_2': V(G) \to V(U_2)$ given by $\phi_2'(v')=\phi_2(v)$ and elsewhere equal to $\phi_2$; note that $\phi_2' < \phi_2$.

We claim that $\phi_2'$ defines a morphism from $G_2$ to $U_2$, which contradicts minimality of $\phi_2$.
There are three cases.
\begin{itemize}
\item Edges in $G_2$ not adjacent to $v'$ are preserved by $\phi'_2$ because they are preserved by $\phi_2$.
\item Let $e = v' \re{c_2} v''$ be an edge outgoing from~$v'$ in~$G_2$.
Then we saw that $v \re {c_2} v''$ is an edge in~$G_2$.
\begin{itemize}
\item If $v'' \neq v'$ then we have
\[
\phi_2'(v') = \phi_2(v) \re {c_2} \phi_2(v'') = \phi_2'(v'') \tin U_2,
\]
the wanted result.
\item If $v'' = v'$ then we have
\[
\phi'_2(v') = \phi_2(v) \re {c_2} \phi_2(v'') = \phi_2(v') \geq \phi_2'(v') \tin U_2
\]
so we conclude thanks to monotone composition in $U_2$ that $\phi'_2(v') \re {c_2} \phi'_2(v') = \phi'_2(v'')$.
\end{itemize}
\item Last, for edges $v'' \re {c_2} v'$ incoming in $v'$ with $v'' \neq v'$ (the case where $v''=v'$ is treated just above), we conclude directly by monotone composition since
\[
\phi'_2(v'') = \phi_2(v'') \re {c_2} \phi_2(v') \geq \phi_2(v') \tin U_2. \qedhere
\]
\end{itemize} 
\end{proof}

Now for each $u_{2} \in V(U_{2})$, we let $G^{u_{2}}_1$ denote the restriction of $G$ to $\phi_2^{-1}(u_{2})$ and to edges with color in $C_1$.
Again, $G_1^{u_{2}}$ is only a pregraph in general, which is not an issue.
(Also it may be empty for some $u_{2}$'s which is not an issue either.)

For each $u_{2} \in V(U_{2})$, the pregraph $G^{u_{2}}_1$ satisfies $W$ since $G$ does, and it even satisfies $W_1$ since it has only edges with color in $C_1$ and $W \cap C_1^\omega = W_1$.
Thus there exists for each $u_{2} \in V(U_{2})$ a morphism $\phi_1^{u_{2}}$ from $G^{u_{2}}_1$ to $U_{1}$.

We now define a map $\psi: V(G) \to V(U)$ by
\[
\psi(v)_1 = \phi_1^{\phi_2(v)}(v) \qquad \tand \qquad \psi(v)_{2} = \phi_2(v).
\]
The following result concludes the proof of Theorem~\ref{thm:universality_lexicographical_product}.

\begin{lemma}{}
The map $\psi$ defines a morphism from $G$ to $U$.
\end{lemma}

\begin{proof}
We have to verify that
\[
v \re c v' \tin G \qquad \implies \qquad \psi(v) \re c \psi(v') \tin U,
\]
and we separate two cases.
\begin{itemize}
\item If $c \in C_2$ then $v \re c v'$ in $G_2$ thus $\psi_2(v) \re c \psi_2(v')$ in $U_{2}$ which yields the result.
\item If $c \in C_1$ we know by Lemma~\ref{lem:no_small_edge_go_up} that $\psi_2(v) \geq \psi_2(v')$.
If this inequality is strict then the $\psi(v) \re {c} \psi{v'}$ by definition of $U$.
Otherwise we have $\phi_2(v)=\phi_2(v')$, and conclude thanks to the fact that $\phi^{\phi_2(v)}$ is a morphism from $G^{\phi_2(v)}$ to $U_{1}$. \qedhere
\end{itemize}
\end{proof}

\section{A class of positional objectives closed under countable unions}\label{sec:healing}

Before giving a high-level overview of this section, we formally introduce directed sums of monotone graphs.
In this section, we will require using strategies for Adam; these are defined just like strategies for Eve modulo inverting the roles of the two players.

\paragraph*{Directed sums.}

Let $(G_\lambda)_{\lambda < \alpha}$ be an ordinal sequence of $C$-monotone graphs.
Their directed sum $G = \sum_{\lambda < \alpha} G_\lambda$ is defined over $V(G) = \sum_{\lambda < \alpha} V(G_\lambda) \times \{\lambda\}$ by
\[
    E(G) = \{((v,\lambda) \re c (v',\lambda)) \mid v \re c v' \tin G_\lambda\} \cup \{(v,\lambda) \re c (v',\lambda') \mid \lambda > \lambda'\}.
\]
In words, we take the disjoint union of the $G_\lambda$'s, and close it by adding all descending edges between different copies.
It is naturally ordered lexicographically, meaning that
\[
    (v,\lambda) \geq (v',\lambda') \quad \iff \quad \lambda > \lambda' \tor [\lambda = \lambda' \tand v \geq v'].
\]
It is a direct check that $G$ is monotone.
Ordinal sums of well-ordered sets are well-ordered lexicographically, thus if the $G_\lambda$'s are well-monotone then so is $G$.

\paragraph*{High-level overview.}

In this section, we elaborate on Kopczy\'nski's conjecture that prefix-independent positional objectives are closed under unions.
It is not clear, given two well-monotone universal graphs $U_1,U_2$ respectively for $W_1$ and $W_2$, how to construct a well-monotone universal graph for $W_1 \cup W_2$ in general.
A naive construction that quickly comes to mind, is to ``alternate and repeat'', formally we consider the graph
\[
    U_\alpha = (U_1 + U_2) \alpha,
\]
for some big enough ordinal $\alpha$.
See Figure~\ref{fig:alternate_and_repeat}.

\begin{figure}[h]
    \begin{center}
    \includegraphics[width=\linewidth]{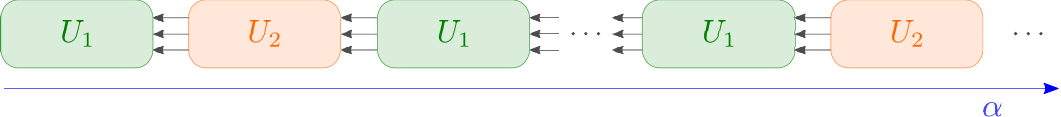}
    \end{center}
    \caption{The graph $U_\alpha = (U_1 + U_2) \alpha$ used as a naive construction for the union of $W_1$ and $W_2$.}\label{fig:alternate_and_repeat}
\end{figure}

Unfortunately, this construction is not well-behaved in general: for instance, if $C=\{a,b\}$ and $W_1=(b^*a)^\omega$ and $W_2=(a^*b)^\omega$ are B\"uchi conditions, then $W_1 \cup W_2 = C^\omega$ therefore it has a 1-vertex (thus well-monotone) uniformly universal graph $U_0$.
However, if $U_1$ and $U_2$ are taken to be the universal graphs discussed previously for B\"uchi conditions, then it is not hard to see that $U_\alpha$ does not embed $U_0$ (and therefore, cannot be $(U_1 \cup U_2)$-universal).

It turns out that there is a wide and natural class $\W$ of positional prefix-independent objectives for which the above construction actually works.
This class is similar to what Kopczy\'nski calls ``positional/suspendable'' conditions in his thesis~\cite{Kopczynski06}; however we give a simpler definition and obtain completely different proofs.

\paragraph*{Healing.}
A prefix of an infinite word $w=w_0w_1 \dots \in C^\omega$ is a finite word $w_0w_1 \dots w_k \in C^*$.
We let $\Pref(w) \subseteq C^*$ denote the set of prefixes of a given infinite word $w$.
Given a condition $W$, we say that a finite word $u \in C^*$ heals an infinite word $w \notin W$ if for all infinite sequences of nonempty prefixes $p_0,p_1, \dots \in \Pref(w)$ of $w$, it holds that $p_0up_1u\dots \in W$.
We insist that in our definition, $p_0,p_1 \dots$ is any sequence of prefixes of $w$; it may or may not be bounded, and it does not hold in general that $w=p_0p_1\dots$.

For example, if $W=(\wait^*\good)^\omega$ is a B\"uchi objective, then the word $u=\good$ heals the word $\wait^\omega \notin W$.
However, for a co-B\"uchi objective $W=(\safe + \bad)^* \safe^\omega$, an infinite word $w \notin W$ cannot be healed by any finite word.
More generally, we say that an objective $W$ admits healing if there is a $w \notin W$ healed by a $u \in C^*$, and otherwise, that it excludes healing.
We will be interested in prefix-independent positional objectives that exclude healing; some examples are given below.
We refer to Section~\ref{sec:examples} for their formal definitions.

\begin{lemma}
The following prefix-independent positional objectives:
\begin{itemize}
\item $\Cobuchi$, and more generally, $K$-monotone objectives
\item $\QEnergy$
\item $\BoundedCounter$
\item $\FinParity$
\end{itemize}
exclude healing.
However, $\Buchi$ admits healing.
\end{lemma} 

We exclude a proof of the lemma which is easy to verify in each case.

\paragraph*{Main result.} Our main result in this section is the following.

\begin{theorem}\label{thm:main_healing}
    Let $W_0,W_1 \dots \subseteq C^\omega$ be a sequence of determined\footnote{An objective is determined if from every vertex, one of the two player has a winning strategy. All objectives considered in this paper are determined.} prefix-independent positional objectives with a neutral letter and which exclude healing, let $\kappa$ be a cardinal number, let $\alpha$ be an ordinal with $|\alpha| \geq \kappa$, and let $U_0,U_1,\dots$ be well-monotone graphs such that $U_i$ is $(\kappa,W_i)$-universal.

    The well-monotone graph $U = (\sum_{i < \omega} U_i) \alpha$ is $(\kappa, \bigcup_{i<\omega} W_i)$-universal.
\end{theorem}

Combining it with Theorem~\ref{thm:main}, we obtain that countable unions of prefix-independent positional objectives which admit a neutral letter and exclude healing are positional.
In contrast, it was proved in~\cite{Kopczynski06} that uncountable unions of co-B\"uchi objectives are not always positional.

\paragraph*{Healing game.}
Toward proving Theorem~\ref{thm:main_healing}, we introduce an auxiliary game $\H_W$ defined for any condition $W \subseteq C^\omega$ with a neutral letter $\eps$ and which we call the $W$-healing game.
It is played in rounds as follows: in the $i$-th round,
\begin{itemize}
\item Eve chooses an infinite word $w_i \notin W$;
\item Adam chooses a nonempty finite prefix $p_i \in \Pref(w_i)$; 
\item Eve chooses a finite word $u_i \in C^*$;
\end{itemize}
then the game proceeds to the next round.
After infinitely many rounds, Eve wins the game if $p_0 u_0 p_1 u_1 \dots \in W$.

Formally, we define $\H_W=(H_W,\VE,W)$ as the game with condition $W$ on the infinite $C$-graph $H_W$ given by 
\[
    V(H_W) = C^* \times \{0,1\} \cup (C^\omega \setminus W),
\]
with only two states controlled by Eve, $\VE = \{(\emptyword,0),(\emptyword,1)\}$, and edges
\[
    \begin{aligned}
    E(H_W)  = & \{(w_0 \dots w_k ,b) \re {w_0} (w_1 \dots w_k,b) \mid k\geq 0, b \in \{0,1\}\} \\
        & \cup \{ (\emptyword,0) \re \eps w \mid w \notin W\} \\
        & \cup \{ w \re \eps (p,1) \mid p \in \Pref(w) \setminus \{ \emptyword\} \} \\
        & \cup \{ (\emptyword,1) \re \eps (u,0) \mid u \in C^* \}.
    \end{aligned}
\]
The game $\H_W$ is depicted in Figure~\ref{fig:healing_game}.

\begin{figure}[h]
    \begin{center}
    \includegraphics[width=0.7\linewidth]{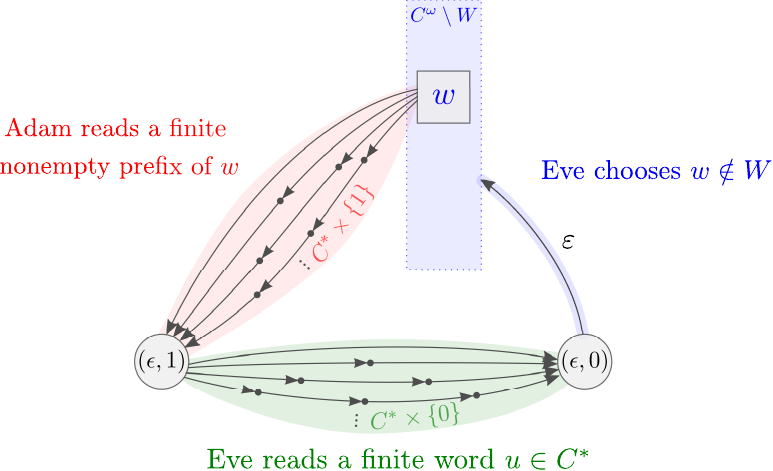}
    \end{center}
    \caption{The healing game $\H_W$.}\label{fig:healing_game}
\end{figure}

\begin{lemma}\label{lem:healing_game}
Let $W$ be a prefix-independent positional condition with a neutral letter.
Then Eve wins $\H_W$ (from any vertex) if and only if $W$ admits healing.
\end{lemma}

\begin{proof}
Assume Eve wins $\H_W$.
Then she wins with a positional strategy $P$.
Consider the two choices of Eve, namely let $w \notin W$ be such that the edge $(\emptyword,0) \re \eps w$ belongs to $P$, and let $u \in C^*$ be such that the edge $(\emptyword,1) \re \eps (u,0)$ belongs to $P$.
Then for any sequence $w_0, w_1, \dots$ of nonempty prefixes of $w$, there is a path of the form
\[
    (\emptyword,0) \rp{\eps w_0} (\emptyword,1) \rp{\eps u} (\emptyword,0) \rp{\eps w_1} (\emptyword,1) \rp{\eps u} \dots
\]
in $P$.
This path has coloration $\eps w_0 \eps u \eps w_1 \dots$, and since $P$ is winning, $w$ is healed by $u$.

Conversely, if $w \notin W$ is healed by $u$, then the positional strategy $P$ obtained by restricting edges outgoing from $\VE$ in $H_W$ to $(\emptyword,0) \re {\eps} w$ and to $(\emptyword,1) \re \eps (u,0)$ is winning: up to removing finite prefixes, all infinite paths have colorations of the form $w_0 u w_1 u \dots$, where $w_0,w_1, \dots$ are finite prefixes of $w$.
\end{proof}

We will use the lemma in the contrapositive: if $W$ excludes healing then Adam wins the healing game, assuming determinacy of $W$.

\paragraph*{Proof of Theorem~\ref{thm:main_healing}.}
Let us repeat the hypotheses of the theorem. 
We let $W_0,W_1,\dots \subseteq C^\omega$ be a sequence of determined prefix-independent positional objectives with neutral letters $\eps_i \in C$ and which exclude healing, let $\kappa$ be a cardinal and let $U_0,U_1, \dots$ be well-monotone graphs such that $U_i$ is $(\kappa,W_i)$-universal.
Let $U=\sum_{i<\omega} U_i$ and let $W=\bigcup_{i<\omega} W_i$.

\begin{lemma}
    The graph $U=\sum_{i < \omega} U_i$ is almost $(\kappa,W)$-universal.
\end{lemma}

This implies Theorem~\ref{thm:main_healing} thanks to Lemma~\ref{lem:rongeur_de_croutes}.
The proof of the lemma hinges on the construction of an infinite path in a graph $G$.
We refer the reader to Figure~\ref{fig:proof_healing} which illustrates and explains the $j$-th step of the construction.

\begin{figure}[h]
    \begin{center}
    \includegraphics[width=\linewidth]{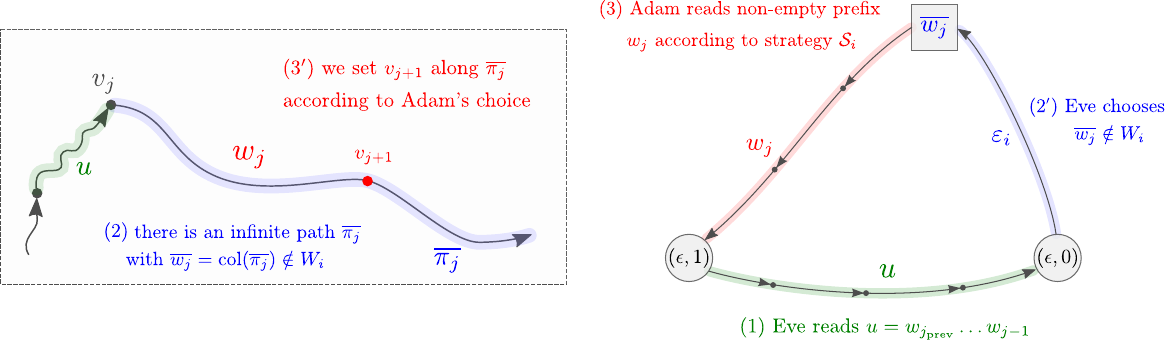}
    \end{center}
    \caption{A depiction of the $j$-th step of the proof. On the left, the graph $G$, and on the right, the game $\H_{W_{i}}$, where $i=e_j$ (see proof below).}\label{fig:proof_healing}
\end{figure}

\begin{proof}
Since each of the $U_i$'s satisfy $W_i$ and since $W$ is prefix-independent, their directed sum~$U$ satisfies $W$ as required.

Let $G$ be a graph of cardinality $< \kappa$ satisfying $W$, and assume towards contradiction that for all $v \in V(G)$ and for all $i \in \omega$, $G[v]$ does not satisfy $W_i$.
Pick any vertex $v_0$; our aim is to construct a path from $v_0$ in $G$ which does not satisfy any of the $W_i$'s, which contradicts the fact that $G \models W$.

By Lemma~\ref{lem:healing_game} and determinacy of the $W_i$'s, Adam wins the healing games $\H_{W_i}$ from all vertices.
For all $i$ we pick winning strategies $\S_i=(S_i,\phi_i,s_{0,i})$ for Adam in $\H_{W_i}$ from\footnote{To avoid overly cluttered notations, we use the same names for vertices in different healing games $\H_{W_i}$. All ambiguities introduced by this choice are easily resolved in context.} $(\emptyword,1)$ (formally, $\phi_i(s_{0,i}) = (\emptyword,1)$).

We construct in parallel an infinite path $\pi$ from $v_0$ in $G$, as well as for each $i$, an infinite path $\pi_i$ from $s_{0,i}$ in $S_i$.
By construction, up to removal of occurrences of neutral letters, all these paths will have the same coloration $w=\col(\pi)$ as $\pi$.
This leads to a contradiction: since $S_i$ is winning for Adam we get $w \notin W_i$ for all $i$, however it colors a path in $G$.

Let $e=e_0e_1 \dots \in \omega^\omega$ be an infinite sequence of integer containing infinitely many occurrences of each $i \in \omega$, for instance $e=001012012301234\dots$.
We construct our paths incrementally; we will have finite nonempty paths $\pi^0,\pi^1,\dots$ such that $\pi=\pi^0 \pi^1 \dots$, and likewise, finite (possibly empty) paths $\pi^0_i,\pi^1_i,\dots$ such that $\pi_i=\pi_i^0 \pi_i^1 \dots$.

We will denote $w_j=\col(\pi^j)$, and construct $v_0,v_1, \dots \in V(G)$ such that $\pi^j$ is a nonempty path from $v_j$ to $v_{j+1}$ in $G$.
Overall, all the $\pi_i^{j}$'s will be starting and ending in $\phi_i^{-1}(\emptyword,1) \subseteq V(S_i)$ (that is, just before Eve chooses a finite word in $\H_{W_i}$).

We now fix $j \geq 0$, assume the paths $\pi^{j'}$ and $\pi_i^{j'}$ constructed for $j'<j$, and let $i=e_j$.
In the $j$-th step, we will only extend $\pi$ as well as the $i$-th path $\pi_i$.
Formally, $\pi_i^j$ is the empty path whenever $i \neq e_j$.
We proceed as follows.
\begin{itemize}
\item Let $j_\prev$ denote the previous occurrence of $i$ in $e$ plus one, or $0$ if there is none.
Let $u=w_{j_\prev} w_{j_\prev +1} \dots w_{j-1} \in C^*$, and let $s \in V(S_i)$ be the last vertex on $\pi_i^{j_\prev}$ if $j_\prev \geq 1$, or $s_{0,i}$ if $j_\prev=0$.
It holds by construction that $\phi_i(s) = (\emptyword,1)$.

Since $\S_i$ is an Adam-strategy and $(\emptyword,1)$ belongs to Eve in $\H_{W_i}$, there is a path of the form
\[
    \tag{1}
    s \rp{\eps_i u} s_1 
\]
in $S_i$, with $\phi_i(s') = (\emptyword,0) \in V(H_{W_i})$.

\item By our assumption on $G$, there exists an infinite path $\li{\pi_j}$ from $v_j$ in $G$ with coloration $\li{w_{j}} \notin W$.
Since $\S_i$ is an Adam-strategy and $(\emptyword,0)$ belongs to Eve in $\H_{W_i}$, there is an edge 
\[
    \tag{2}
    s_1 \re {\eps_i} s_2
\]
in $S_i$, with $\phi_i(s_2)=\li {w_j}$.

\item Now consider an outgoing edge $s_2 \re {\eps_i} s'$ in $S_i$; since $\phi_i$ is a morphism it satisfies $\phi_i(s') = (p,1)$ for some nonempty prefix $p$ of $\li{w_j}$.
We set $w_j=p \in \Pref(\li{w_j})$.
There is a path
\[
    \tag{3}
    s_2 \re{\eps_i} s' \rp{w_j} s_3
\]
in $S_i$, with $\phi_i(s_3) = (\emptyword,1)$.
\item We set $\pi^j$ to be the prefix of $\li{\pi_j}$ with coloration $w_j$; it defines a non-empty finite path 
\[
     \pi^j : v_j \rp{w_j} v_{j+1}
\]
in $G$.
\item We set $\pi_i^j : s \rp{} s_3 \in S_i$ to be the concatenation of the paths $(1),(2)$ and $(3)$ above.
Since $\phi_i(s_3) = (\emptyword,1)$, the required invariant is satisfied.
Moreover, we have 
\[
    \col(\pi_i^j) = \eps_i u \eps_i^2 w_j,
\]
which is equal, up to removal of some occurrences of $\eps_i$, to $w_{j_\prev} w_{j_\prev +1} \dots w_{j-1} w_j$.
Thus is holds by induction that, up to removal of some occurrences of $\eps_i$, the coloration of $\pi_i^j$ matches that of $\pi_i$.
\end{itemize}
This concludes the construction of the paths, and leads to the wanted contradiction: $w=w_0w_1\dots=\col(\pi) \in W$ since it colors a path in $G$, however it cannot belong to any $W_i$ since it colors (up to removal of some neutral letters) a path in each $S_i$ (which are winning for Adam).

Therefore there exists $v \in V(G)$ such that $G[v] \models W_i$.
By universality of $U_i$, we then get $G[v] \to U_i \to U$, which defines the wanted morphism.
\end{proof}

\section{Conclusion}

In this paper, we introduced well-monotone graphs, and established that existence of universal such graphs, for a given valuation (or objective) characterizes its positionality.
Our proof of the implication from positionality to existence of universal structures requires a neutral color; we do not know whether this imposes a restriction on the class of positional valuations.

\paragraph*{A method for proving positionality.}
We believe that our work provides a very efficient tool for proving positionality of various valuations or objectives.
Experience has shown that, given, say, an objective $W \subseteq C^\omega$, one may usually approach proving positionality of $W$ as follows.
First, look for a candidate universal well-monotone graph $U$.
Then try to prove its universality (or often, we rather aim for almost universality which is sufficient thanks to Lemma~\ref{lem:rongeur_de_croutes}).
If this does not succeed, then one finds a graph $G$ that does not embed in $U$; this usually leads to two possible outcomes:
\begin{itemize}
    \item either we can enrich $U$ or slightly alter its structure in order to embed $G$, and then continue the process;
    \item or $G$ leads to a counter example to $W$'s positionality, for instance by constructing the game~$\cG'$ obtained from~$G$ as in Section~\ref{subsec:positionality_implies_structure}.
\end{itemize}

Let us illustrate this method on an example.
Let $C=\omega$, and for $w \in \omega^\omega$, let $\inf(w) \subseteq \omega$ denote the set of colors that have infinitely many occurrences in $w$.
We consider the objective
\[
    W = \{w \in \omega^\omega \mid \inf(w) \text{ is finite}\}.
\] 
A reasonable first candidate for (almost) universality of $W$ is the graph $U$ over $V(U)=\omega$ given by
\[
    u \re n u' \tin U \quad \iff \quad u>u' \tor [u=u' \tand n \leq u].
\]
See Figure~\ref{fig:first_candidate}.

\begin{figure}[h]
    \begin{center}
    \includegraphics[width=0.3\linewidth]{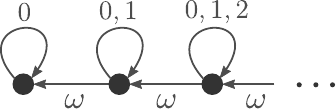}
    \end{center}
    \caption{The graph $U$, a first candidate for universality.}\label{fig:first_candidate}
\end{figure}

We then aim to prove that $U$ is (uniformly) almost universal: for this, one should prove that in a (small) graph $G$ satisfying $W$, there is a vertex $v \in V(G)$ and a bound $N \in \omega$ (corresponding to a valid position for $v$ in a mapping towards $U$) such that on a path from $v$ in $G$, there are at most $N$ occurrences of a color $> N$.

Assume that this is not the case: for all $v \in V(G)$ and for each $N \in \omega$, there is a path from~$v$ with $>N$ occurrences of colors $>N$.
Does this contradict the fact that $G \models W$?
Actually, it does not: the graph $G$ over $V(G)=\omega$ with edges
\[
    E(G) = \{n \re {n} n+1 \mid n \in \omega\}
\]
satisfies $W$ (colorations $w$ on $G$ satisfy $\inf(w)=\emptyset$) but does not meet the above requirement: it contradicts almost universality of $U$.

\begin{figure}[h]
    \begin{center}
    \includegraphics[width=0.75\linewidth]{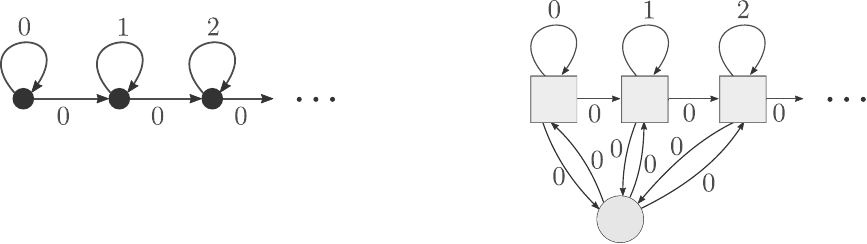}
    \end{center}
    \caption{On the left, the difficult graph $G$. On the right, the game $\game'$ obtained from $G$, and contradicting the positionality of $W$.}\label{fig:difficult_graph}
\end{figure}

At this point, we aim to update our candidate universal graph so that it should embed our new ``difficult graph'' $G$.
However, this fails in this case; the reason is that our $G$ can in fact be turned into a counter example to $W$'s positionality simply by mimicking the gadget from Section~\ref{subsec:positionality_implies_structure} as follows.

Let $\game'=(G',\VE',W)$ be the game with condition $W$ played over $G'$ with $V(G') = \omega \cup \{\bot\}$, $\VE=\{\bot\}$, and
\[
    E(G') = \{n \re n n+1 \mid n \in \omega\} \cup \{n \re 0 \top \mid n \in \omega\} \cup \{\top \re 0 n \mid n \in \omega\}.   
\]
By playing the back-and-forth strategy (when Adam goes to $\top$, go back to the same place), Eve ensures that the coloration $w$ satisfies $\inf(w) \subseteq \{0\}$ and thus wins the game.
However, Adam can beat any positional strategy by playing in a round robin fashion.
This concludes our discussion about $W$: it is not positional.

\paragraph*{A wealth of applications.}
We applied the above method in Section~\ref{sec:examples} to provide positionality proofs for a number of well-studied objectives.
For a few of them (namely, $\QEnergy,\BoundedCounter$ and $\FinParity$), the positionality result we derived is novel.
Going further, we provided two generic constructions for combining together universal graphs for prefix-independent objectives: finite lexicographical products in Section~\ref{sec:lexico}, and countable unions of objectives which exclude healing in 
Section~\ref{sec:healing}.
Despite being introduced only very recently in~\cite{Ohlmann22}, well-monotone graphs have already been exploited by Bouyer, Casares, Randour and Vandenhove~\cite{BCRV22} to characterize positional objectives recognized by deterministic B\"uchi automata.

\paragraph*{Open problems.}
The most tantalizing open question remains Kopczy\'nski conjecture: are unions of prefix-independent positional objectives positional?
This was recently answered in the negative in the case of finite game graphs by Kozachinskiy~\cite{Kozachinskiy22}; however the question remains open in the setting of infinite game graphs considered in this paper.
We believe that well-monotone graphs provide a nice tool to attack this question: on one hand, one can look for constructions combining well-monotone graphs to preserve unions (this is achieved in Section~\ref{sec:healing} assuming healing is excluded), and on the other hand, graphs that are ``hard to embed'' can provide counterexamples (see above for instance).

Other open problems include understanding the role of neutral letters (is it the case in general that if $W$ is positional then so is $W^\eps$? is it the case assuming $W$ is prefix-independent?), and characterizing positionality for $\omega$-regular objectives.
Another direction which we leave to future work is to extend the technology of well-monotone graphs to the case of finite-memory strategies; presently, no characterization result is known for objectives (or valuations) which are finite-memory determined for Eve over infinite graphs.

\printbibliography

\end{document}